%% file: main.tex
\documentclass[a4paper,UKenglish]{lipics}
\usepackage{microtype}%if unwanted, comment out or use option "draft"
\usepackage{breakurl}             % Not needed if you use pdflatex only.
\usepackage[fleqn]{amsmath}
\usepackage{amssymb,latexsym,wasysym,dsfont,stmaryrd,mathtools}
\usepackage{relsize}
\usepackage{xspace}
\usepackage{tikz}
\usetikzlibrary{arrows,automata,shapes.geometric,decorations.pathmorphing,backgrounds,positioning,fit,petri}
\usepackage{wrapfig}
\usepackage{mymacros,addmacros}
\usepackage{calc}
\usepackage{manfnt}
\usepackage{sectionningmacros} % For additional theorem environments

\newif\ifdraft\draftfalse
\ifdraft
\newcommand\modedraft[1]{#1}

\newcommand\todo[1]{{\color{purple}
[\textbf{To do:} #1]}}
\newcommand\cd[1]{\marginpar[{\color{cyan}\small\dbend}]{\color{cyan}\small\dbend}{\footnotesize \color{cyan}[#1 - \textbf{Catalin}]}}
\newcommand\bm[1]{\marginpar[{\color{blue}\small\dbend}]{\color{blue}\small\dbend}{\footnotesize \color{blue}[#1 - \textbf{Bastien}]}}
\renewcommand\sp[1]{\marginpar[{\color{orange}\small\dbend}]{\color{orange}\small\dbend}{\footnotesize \color{orange}[#1 - \textbf{Sophie}]}}
\newcommand{\bmchanged}[1]{{\color{blue}{#1}}}
\newcommand\cdchanged[1]{{\color{cyan}{#1}}}
\newcommand\spchanged[1]{{\color{orange}{#1}}}
\else
\newcommand\modedraft[1]{}
\newcommand\todo[1]{}
\newcommand\cd[1]{}
\newcommand\bm[1]{}
\renewcommand\sp[1]{}
\newcommand\cdchanged[1]{#1}
\newcommand\bmchanged[1]{#1}
\newcommand\spchanged[1]{#1}

\fi
\title{The Expressive Power of Epistemic $\mu$-Calculus}
\author[1]{C\u at\u alin Dima}
\affil[1]{Universit\'e Paris Est, LACL (EA 4219), UPEC,  Cr\'eteil,
  France --
  \texttt{dima@u-pec.fr}}
\authorrunning{C. Dima} %mandatory. First: Use abbreviated first/middle names. Second (only in severe cases): Use first author plus 'et. al.'

\author[2]{Bastien Maubert}
\affil[2]{IRISA, Universit\'e de Rennes 1,
  Rennes, France --
  \texttt{bastien.maubert@irisa.fr}}
\authorrunning{B. Maubert} %mandatory. First: Use abbreviated first/middle names. Second (only in severe cases): Use first author plus 'et. al.'

\author[3]{Sophie Pinchinat}
\affil[3]{IRISA, Universit\'e de Rennes 1,
  Rennes, France --
  \texttt{sophie.pinchinat@irisa.fr}}
\authorrunning{C. Dima, B. Maubert and S. Pinchinat} %mandatory. First: Use abbreviated first/middle names. Second (only in severe cases): Use first author plus 'et. al.'

\Copyright{C\u at\u alin Dima, Bastien Maubert and Sophie Pinchinat}%mandatory, please use full first names. LIPIcs license is "CC-BY";  http://creativecommons.org/licenses/by/3.0/

% mandatory: Please choose ACM 1998 classifications from http://www.acm.org/about/class/ccs98-html . E.g., cite as "F.1.1 Models of Computation". 
% mandatory: Please provide 1-5 keywords
\serieslogo{}%please provide filename (without suffix)
\volumeinfo%(easychair interface)
  {Billy Editor and Bill Editors}% editors
  {2}% number of editors: 1, 2, ....
  {Conference title on which this volume is based on}% event
  {1}% volume
  {1}% issue
  {1}% starting page number
\EventShortName{}
\DOI{10.4230/LIPIcs.xxx.yyy.p}% to be completed by the volume editor
\begin{document}

\maketitle

\input{abstract}
\input{introduction-variante}

\input{prelim}

\input{equivalence}

\input{aux-inexpressive}

\input{conclusion}

\bibliography{biblio,biblio-suppl}

\newpage

%\end{document}
\appendix
\input{appendix}
\end{document}

%% file: abstract.tex
\begin{abstract}
While the $\mu$-calculus notoriously subsumes Alternating-time Temporal
Logic (ATL), we show that the \emph{epistemic} $\mu$-calculus does not
subsume ATL \emph{with imperfect information} (\ATLi), for the
synchronous perfect-recall semantics.
To prove this we first establish that  jumping parity tree  automata (JTA), a recently introduced extension
of alternating parity tree automata, are
expressively equivalent to the epistemic $\mu$-calculus,  and this for any knowledge semantics. %  This inexpressibility result entails 
%  that the class of jumping automata is not closed under projection and that
% the monadic second-order logic with equal-level predicate is strictly more 
% expressive than the epistemic mu-calculus. 
Using this result we also show that, for bounded-memory semantics, the epistemic
$\mu$-calculus is not more expressive than the standard $\mu$-calculus,
and that its satisfiability problem is \EXPTIME-complete. 
\end{abstract}

%% file: introduction-variante.tex
\section{Introduction}

The propositional \mucalc ($\Lmu$) \cite{DBLP:journals/tcs/Kozen83} is
\cdchanged{a logic of utmost importance} in theoretical computer science for several
main reasons. First, it is a powerful logic that captures all
$\omega$-regular properties that are used for the verification of
\bmchanged{dynamic systems' behavioral properties}. In particular, it subsumes
all classic temporal logics, such as \LTL, \CTL and $\CTLs$
\cite{emerson1990handbook}. Second, it enjoys deep connections with
several  paradigms that play a fundamental role in modern
approaches for the verification of reactive systems: it is equivalent
to alternating parity automata \cite[Chap.\
10]{gradel2002automata}, a powerful tool to design decision procedures
for temporal logics. $\Lmu$ is also closely related with parity games,
which are central both for modeling the interaction of systems and for
testing the satisfiability of temporal logics
\cite{{gradel2002automata}}.
\bmchanged{It can be used to specify strategic abilities in multi-player games \cite{DBLP:conf/atva/Pinchinat07},
and it subsumes % (in its ``modal'' variant)
logics of coalition and strategy like the 
Alternating-time Temporal Logic (ATL) \cite{alur2002alternating} and Strategy Logic \cite{ChatterjeeHP10strategy-logic}.}
Finally, its connection with more classic
logics is well understood as its expressive power coincides with the bisimulation
invariant fragment of the monadic second order logic (MSO) \cite{janin1996expressive}.

%Extending $\Lmu$ with quantifications over propositions does not add
%any expressivity, so that the logic can also be used to specify
%strategic abilities in multi-player games
%\cite{DBLP:conf/atva/Pinchinat07} -- a natural abstraction of e.g.\
%distributed computing systems. The computational complexity of this
%extension is non-elementary, justifying  investigations for
%designing hopefully ``cheaper'' logics for strategic reasoning, one of
%the first and
%most successful being Alternating-time Temporal Logic (ATL)
%\cite{alur2002alternating}; noticeably, $\Lmu$ subsumes ATL \bm{cite}.

While most \bmchanged{results concern} the perfect information setting
in which players/agents know the actual state of the system, realistic
applications led to consider  agents that have to strategize based
on a partial information of their environment.
% However, it is well known that multi-player games with imperfect
% information are intrinsically difficult to analyse [Reif].
This need % need to find tractable solutions
 gave rise to
a proliferation of frameworks to represent, reason about \bmchanged{and/or} strategize under
imperfect information.
% investigating on
% imperfect information is guided by realistic applications
% where individuals strategize upon a partial information of the whole
% picture. However, as established in [Reif], multi-player games with imperfect
% information are intrinsically difficult to analyse.
% Following this need, many proposals from the literature have
% emerged.
There are basically two trends. One trend relies on
extensions of previous strategic logics with additional constraints on
strategic abilities of players, that forces them to strategize
consistently with their available information.
\cdchanged{This is the case of variants of ATL with imperfect information like 
$\ATLi$, $\text{ATL}_{ir}$, $\text{ATLK}$ or $\text{ATEL}$
\cite{jamroga2006agents,wiebe-ATEL2003,schobbens-atl-ir2004} -- to cite only a few, 
see also \cite{bulling-atl-survey} for a recent survey of the various logics of this type}. 
The other trend is based on extensions of temporal logics
with epistemic features, sometimes also combined with the
concepts of the former. 
\bmchanged{Such logics include Epistemic Temporal Logic \cite{halpern1989complexity},  epistemic mu-calculus $\Kmu$, 
first introduced in \cite{shilov-garanina-fixpoints}, and}  \cdchanged{the epistemic alternating mu-calculus AMC 
\cite{bulling-jamroga-mu}}.

Comparing the two trends is necessary to share expertise, and it is
relevant to wonder whether $\Kmu$ has the same central position as the
standard \mucalc has in the perfect information setting.  Some results
are already known: the epistemic \mucalc subsumes Epistemic Temporal
Logic and Propositional Epistemic Dynamic Logic
\cite{shilov-garanina-fixpoints}, and a notion of Alternating
Epistemic Mu-Calculus that considers one-step strategic abilities
\cite{bulling-jamroga-mu}.  \bmchanged{It is also known from
  \cite{bulling-jamroga-mu} that $\ATLi$ is not subsumed by this
  Alternating Epistemic Mu-Calculus for a memoryless semantics with
  imperfect information.}

\cdchanged{Our contribution is threefold: 
first, we show that the epistemic \mucalc has the same expressive power as the 
recently introduced \emph{jumping  automata}, an extension of alternating parity tree automata that
allow for jumps between tree nodes \cite{maubertFSTTCS2013}.
The proof relies on the classic  result  that
the modal \mucalc is equivalent with alternating tree automata  \cite{gradel2002automata}
%\cite{walukiewicz-vectorial-mu}. 

Second, combining this general result with
\bmchanged{the fact that  jumping automata equipped with
\emph{recognizable} relations between tree nodes  translate in
linear time into two-way tree automata
 \cite{maubertFSTTCS2013}, we obtain two
corollaries: for bounded memory semantics, (1) $\Kmu$ is
not more expressive than $\Lmu$, and (2) the satisfiability
problem for $\Kmu$  is
\EXPTIME-complete.}

%  \cdchanged{\cite{maubertFSTTCS2013}}, 
% the latter of which have an \EXPTIME-complete emptiness problem \cite{DBLP:conf/icalp/Vardi98}.

Third, we prove that, unlike in the perfect information setting,
\bmchanged{$\ATLi$ is not subsumed by the epistemic \mucalc}: 
 we consider the formula 
 $\stratA[a]\F p $, which means that Alice has a uniform strategy (\ie
a strategy consistent with her observations) to
eventually reach $p$, 
and we show that if a jumping automaton accepts all the 
(tree) models of this formula
then it also accepts 
another model in which  Alice only has a non-uniform strategy to achieve $\F p$.
This result is proved for the synchronous and perfect recall semantics of indistinguishability.}
%We comment in the conclusions
%As a corollary of this, we may get that the monadic second order logic over trees, extended with 
%the equal-level predicate (see e.g. \cite{thomas-msoeqlevel}) is strictly more expressive than the epistemic \mucalc.
%Due to space constraints, we only mention this result in the conclusions section, leaving its statement and proof 
%or the full version of this paper. }

%Our contribution enlightens the power of both the epistemic \mucalc
%and $\ATLi$ by establishing that the most simple reachability-like
%statement in $\ATLi$ cannot be expressed in $\Kmu$ for the synchronous
%perfect recall semantics. To attain this result, we first characterize
%the expressive power of the epistemic \mucalc, for arbitrary
%semantics, in automata-theoretic terms. We prove that $\Kmu$ formulas
%are translatable in linear time into so-called \emph{jumping
%  automata}, an extension of alternating parity tree automata that
%allow for jumps between tree nodes. Combining this general result with
%previous developments on jumping automata equipped with
%\emph{recognizable} relations between tree nodes \sp{cite}, gives us two
%corollaries: (1) For bounded memory semantics, $\Kmu$ is
%not more expressive than $\Lmu$. (2) The satisfiability
%problem for $\Kmu$ is \EXPTIME-complete, since jumping automata then
%translate in linear time into two-way alternating parity automata \sp{cite}, the latter of which have an \EXPTIME-complete emptiness problem \sp{cite}.

The paper is organized as follows. In Section~\ref{sec-prelim}, first we
 introduce basic notations and we recall classic parity games
as well as game bisimulations. We then expose the
epistemic \mucalc, ATL with imperfect information, and jumping tree
automata. In Section~\ref{sec-equivalence} we prove that the epistemic
\mucalc is equivalent to jumping tree automata, from which we derive 
corollaries on the expressivity and the complexity of $\Kmu$ with
bounded memory. Using again the correspondence between $\Kmu$ and jumping
tree automata, we prove in Section~\ref{sec-inexp} that ATL with
imperfect information is not expressible in $\Kmu$, and we conclude in Section~\ref{sec-conclusion}, 
where we also comment on the impact of the results on the relationship between the epistemic \mucalc
and the monadic second order enriched with equal-level predicate (see e.g. \cite{thomas-msoeqlevel}).

%% file: prelim.tex
\section{Preliminaries}
\label{sec-prelim}

In this section we set some notations concerning infinite trees and
parity games, and we recall the definitions of the three main objects
considered in this paper: \spchanged{epistemic \mucalc, ATL with imperfect
information, and jumping tree automata}.

% \subsection{Infinite trees}
% \label{sec-trees}

\newcommand{\APtree}{$\APf$-tree\xspace}

% For the rest of the paper, let $d\in \setn$ be a natural number (the
% maximum \emph{branching degree}), 

%For $d\in\setn$, we let $[d]\egdef\{0,\ldots,d-1\}$. 
A
\emph{tree} is a nonempty set $\tree\subseteq \setn^*$%  (for
% some $d$ called the \emph{maximum branching degree})
such that
if $x\cdot i\in \tree$, then $x\in \tree$ \spchanged{and $x.j \in \tree$ for all $j<i$}, and if $x \in \tree$, there
exists $i\in\setn$ such that $x\cdot i\in\tree$, and if $x.i$. The elements of $\tree$ are called \emph{nodes},
and the empty word $\epsilon$ is the \emph{root} of the tree. If
$x\cdot i \in \tree$,  $x\cdot i$ is a \emph{child} of
$x$. %\sp{``parent'' is unused}, and  $x$ is the \emph{parent} of $x\cdot i$.
The \emph{arity} of a node $x$
% $\arity(x)$,
is its number of children, and if every node of some tree
$\ltree$ has arity at most $k$, % \spchanged{(with value $k$ reached)},
$\tree$ is a \emph{$k$-ary tree}. % We will
% note $x\cdot \up$ the parent of a node $x$:  $(x\cdot i)\cdot
% \up\egdef x$. Note that $\epsilon \cdot \up$ is undefined as the
% root has no parent.
% The \emph{arity} of a node $x$, written
% $\arity(x)$, is the number of children of $x$. If all nodes have
% arity $k$, then $\tree$ is a $k$-ary tree. \fbox{Enlever definitions inutiles}
Given a node $x$ of a tree $\tree$, we let \spchanged{$\Paths_\tree(x)$ (or simply $\Paths(x)$)}  be the set of
infinite paths $\pi=x_0x_1\ldots$ in $\tree$ such that $x_0=x$ and for
all $i$, $x_{i+1}$ is a child of $x_i$. Also, for a path
$\pi=x_0x_1\ldots$ we let $\pi[i]:=x_i$.  For two nodes $x$ and $y$,
$y$ is a \emph{descendent} of $x$ (written $x\desc y$) \spchanged{if
  $x$ is a prefix of $y$, or equivalently} if $y$ can be found on some
path that starts in $x$. \spchanged{We denote by} $\subt{\tree}{x}$ the subtree of
$\tree$ rooted in $x$: $\subt{\tree}{x}=\{y\mid x\desc y\}$.

Trees may be \labeled with \emph{atomic propositions}
from a countably infinite set $\AP$  that we fix.
For a finite subset $\APf\subset\AP$ of atomic propositions, an \emph{\APtree} is a pair
$\ltree=(\tree,\lab)$, where $\tree$ is a tree and $\lab:\tree \to
2^\APf$ is a \emph{labelling} \spchanged{of the nodes}. 
\spchanged{A node $x$ in a tree is reached by a finite prefix $\rho$ of a path in $\Paths(\epsilon)$, say $\rho_x=x_0 \ldots x_n$ with $x_n=x$. 
We define the \emph{word} of $x$, written $\word(x)$, by $\lab(\epsilon)\lab(x_1)\ldots\lab(x_n)$.}

% For a node $x=i_1i_2\ldots i_n$ in $\tree$, $n\geq 0$, we define its
% \emph{node word} $\word(x)$ made of the sequence of labels from the
% root to this node: $\word(x) \egdef
% \lab(\epsilon)\lab(i_1)\lab(i_1i_2)\ldots \lab(i_1\ldots i_n)$. 

For
simplicity, we may write $x\in\ltree$ instead of $x\in\tree$. Finally,
if $\ltree=(\tree,\lab)$ is an \APtree, $p\in\AP$ and $S\subseteq
\tree$, we define $\ltree[p\to S]$ as the $(\APf\union \{p\})$-tree
$\ltree'=(\tree,\lab')$, where $\lab'(\noeud)=\lab(\noeud)\union\{p\}$
if $\noeud\in S$, and $\lab(\noeud)\setminus\{p\}$ otherwise. In other
words, $\ltree[p\to S]$ is the same tree as $\ltree$, except that we
make $p$ hold exactly on nodes in $S$.

% The notation $w\up$ is defined identically for node words as $x\up$
% is for nodes. Notice that $\word(\epsilon)=\lab(\epsilon)$, hence $\word(\epsilon)\up=\epsilon$.
%\pagebreak
% \begin{definition}
% \label{def-trees}
% Let $n\in\setn$, and let $\Sigma$ be a finite alphabet. An
% \emph{$n$-ary $\Sigma$-tree} is a partial mapping
% $\tree:\{1,\ldots,n\}^*\to \Sigma$ whose domain, $\dom(\tree)$, is prefix-closed.
% \end{definition}

% For an $n$-ary tree $\tree$, we call \emph{node} an element $\noeud$ of $\dom(\tree)$. A
% node $\noeuda$ is a \emph{child} of a node $\noeud$ if
% $\noeuda=\noeud\cdot k$ for some $k\in\{1,\ldots,n\}$. In that case,
% $\noeud$ is called the unique \emph{parent} of $\noeuda$. The empty
% word $\racine$ is called the \emph{root} of the tree.

\subsection{Parity games and game bisimulation}
\label{sec-games}

We define two-player turn-based parity games, that we use
to define acceptance of trees by parity tree automata. \bmchanged{We also
define  game bisimulations, recently introduced in
\cite{DBLP:journals/jolli/BerwangerK10}.} %But first a couple of notations: for an infinite word $w\in \Sigma^\omega$ where $\Sigma$ is some alphabet, 
% if $w=a_0a_1\ldots$, then for every $i\geq 0$, we let $w[i]\egdef
% a_i$ and $w[0,i]\egdef a_0a_1\ldots a_i$. Also, for a finite word
% $w= a_0\ldots a_{n-1}$, we let $|w|\egdef n$ denote its \emph{length}.

\spchanged{Fix an alphabet $\Sigma$. For an infinite word
  $w=a_0a_1\ldots \in \Sigma^\omega$ \bmchanged{and $i\geq 0$}, we let $w[i]\egdef a_i$ and
  $w[0,i]\egdef a_0a_1\ldots a_i$. For a finite
  word $u= a_0\ldots a_{n-1} \in \Sigma^*$, its \emph{length} is
  $|u|\egdef n$.}

\newcommand{\gav}[1]{\ga_{#1}}
%\sp{pour moi clash de notations entre les paths d'un arbre et les play des jeux}

% \begin{definition}
% \label{def-games}
We define two-player turn-based parity games: A \emph{parity game
  arena} is a tuple $\ga=(V,E,C)$, where $V$ is a set of
\emph{positions} partitioned between positions of Eve ($V_E$) and
those of Adam ($V_A$). \bmchanged{Binary relation} $E\subseteq V\times
V$ is a set of \emph{moves} that we assume total, \ie for all $v\in
V$, there is $v'\in V$ such that $(v,v')\in E$.  Finally, $C:V\to
\setn$ is a \emph{colouring function}.  A \emph{parity game}
$\game=(\ga,v_0)$ is a game arena $\ga=(V,E,C)$ together with an
initial position $v_0\in V$.  Given a parity game $\game=(\ga,v_0)$, a
\emph{play} $\pi\in V^\omega$ is an infinite sequence of positions
such that $\pi[0]=v_0$, and for all $i\geq 0$, $(\pi[i],\pi[i+1])\in
E$. A \emph{partial play} $\rho=v_0 \ldots v_n\in V^*$ is a finite
prefix of a play \spchanged{and it \emph{ends} in $v_n$}. A
\emph{strategy} $\strat$ for Eve is a partial function $\strat:V^*\to
V$ such that for all partial play $\rho$ ending in $v\in V_E$,
$\strat(\rho)$ is defined and $(v,\strat(\rho))\in E$. A play $\pi$
\emph{follows} a strategy $\strat$ if for all $i\geq 0$ such that
$\pi[i]\in V_E$, $\pi[i+1]=\strat(\pi[0,i])$, and similarly for
partial plays.  For a parity game $\game$ and a strategy $\strat$ for
Eve in $\game$, we denote by $\out(\game,\strat)$ \spchanged{the set
  of \emph{outcomes} of $\strat$, that is plays in $\game$ that follow
  $\strat$}. A play $\pi$ is \emph{winning} for Eve if the least
colour seen infinitely often along $\pi$ is even, otherwise $\pi$ is
winning for Adam. A \emph{winning strategy} for Eve is a strategy
whose outcomes are all winning for Eve.  Finally, \bmchanged{as we only consider
winning strategies of Eve,} we say that position $v$ of a game arena
$\ga$ is \emph{winning} if \spchanged{Eve} has a winning strategy in
$(\ga,v)$.
%\end{definition}

Berwanger and Kaiser introduce in
\cite{DBLP:journals/jolli/BerwangerK10} a notion of bisimulation
between parity games and they prove that two bisimilar games are
equivalent with regards to the existence of winning strategies
\footnote{Note that in \cite{DBLP:journals/jolli/BerwangerK10} the
  definitions are more general and consider games with imperfect information.}. This
result will be crucial to establish our nonexpressivity result in
Section~\ref{sec-inexp}.

\begin{definition}
\label{def-game-bisim}
Let $\ga=(V,E,C)$ and $\ga'=(V',E',C')$ be two game arenas. A
\emph{bisimulation} between $\ga$ and $\ga'$ is a binary relation
$Z\subseteq V\times V'$ such that:
\begin{description}
  \item[Colour Harmony:] for all $(v,v')\in Z$, $C(v)=C'(v')$,
  \item[Zig:] for all $(v,v')\in Z$, if there is $u\in V$ such that
    $(v,u)\in E$, then there is $u'\in V'$ such that $(v',u')\in E'$
    and $(u,u')\in Z$, and
  \item[Zag:] for all $(v,v')\in Z$, if there is $u'\in V'$ such that
    $(v',u')\in E'$, then there is $u\in V$ such that $(v,u)\in E$
    and $(u,u')\in Z$.
\end{description}
For initial positions $v_0\in V$ and $v_0'\in V'$, we say that
$(\ga,v_0)$ is \emph{bisimilar} to $(\ga',v_0')$, \spchanged{written $\ga,v_0 \bisim \ga',v_0'$}, if there is a bisimulation
$Z$ between $\ga$ and $\ga'$ such that $(v_0,v_0')\in Z$. 
\end{definition}

% Berwanger and Kaiser prove that if two games are bisimilar then they
% are equivalent in a sense that they define. Here we do not need the
% full power of their result, so we only state a weaker
% form.

\begin{proposition}[\cite{DBLP:journals/jolli/BerwangerK10}]
\label{prop-game-bisim}
For two game arenas $\ga$ and $\ga'$, and two respective
  positions $v$ and $v'$, if $\ga,v\bisim\ga',v'$, then $v$ is winning
  in $(\ga,v)$ if and only if $v'$ is winning in $(\ga',v')$.
\end{proposition}

\subsection{Epistemic \mucalc }
\label{sec-mucalc}

We fix $\setvarsec=\{X,Y,\ldots\}$ a countably infinite set of
 \emph{second
  order variables}. Given  a finite set of
\emph{agents} $\agents$, the syntax of the epistemic \mucalc $\Kmu$ is defined by
the following grammar:
\[\phi ::= X \mid p \mid \neg \phi \mid \phi \ou \phi \mid \diam
\phi \mid \Ki \phi \mid \mu X.\phi(X)\]
where  $X\in \setvarsec$, $p\in\AP$, $i\in\agents$, and in the last rule $X$ appears
only positively (under an even number of negations) in $\phi(X)$.
For a \emph{finite} set of atomic propositions $\APf\subset \AP$, we denote by $\Kmu(\APf,\agents)$, or simply $\Kmu$ when the
parameters are irrelevant, the set of formulas of the epistemic
\mucalc that only use atomic propositions in $\APf$ and agents in $\agents$.

A model of a formula in $\Kmu(\APf,\agents)$ consists in an
$\APf$-tree $\ltree$
together with a set of binary relations $\{\reli\}_{i\in\agents}$
over $(2^\APf)^*$. % For each $i\in\agents$, $\reli\;\subseteq (2^\APf)^*
% \times (2^\APf)^*$
% represents Agent $i$'s
% observational and memorial abilities.
% Formally, we call \emph{epistemic $\Sigma$-tree} a
% tuple $\Ktree=(\ltree,\{\reli\}_{i\in\agents})$, where $\ltree=(\tree,\lab)$ is a
% $\Sigma$-tree, and for each $i\in\agents$, $\reli\subseteq
% \ltree\times\ltree$ is an \emph{epistemic relation} for agent $i$.
In the following, for two nodes $\noeud$ and $\noeuda$ in $\ltree$,
$\noeud\reli\noeuda$ stands for $\word(\noeud)\reli\word(\noeuda)$:
two nodes are related by $\reli$ if their node words  are related by $\reli$.
Intuitively,
$\noeud\reli\noeuda$ means that when the current node is $\noeud$, Agent
$i$ considers possible (up to her knowledge) that node $\noeuda$ is
the current node.
Notice that the relation $\reli$ is arbitrary and not necessarily an
equivalence relation, as often assumed in epistemic logic.
From now on, whenever $\agents$ is clear from the context, $\setrel$ will
denote a \emph{relation profile} $\{\reli\}_{i\in\agents}$.
Finally, interpreting a formula  requires a
\emph{valuation}
$\val:\setvarsec\to 2^{\ltree}$; also, given $X\in\setvarsec$ and
$S\subseteq t$, $\val[S/X]$ is the valuation that maps $X$ to $S$, and
is equal to $\val$ on all other variables. 

The semantics of a formula $\phi\in\Kmu(\APf,\agents)$ on an
$\APf$-tree $\ltree=(\tree,\lab)$
with relation profile $\setrel$ and valuation $\val$ is the set of
nodes  $\sem{\phi}\subseteq\ltree$
 defined as follows:
\[
\begin{array}{clcl}
  \bullet & \sem{X}=\val(X) & \hspace{2cm}
  \bullet & \sem{p}=\{\noeud\in\ltree\mid p \in \lab(\noeud)\}\\[3pt]
  \bullet &  \sem{\neg\phi}=\ltree \sminus \sem{\phi} & \hspace{2cm}
  \bullet &  \sem{\phi\ou\psi}=\sem{\phi} \union \sem{\psi}\\[3pt] 
  \bullet &  \multicolumn{3}{l}{\sem{\diam\phi}=\{\noeud\in\ltree \mid \noeud\cdot i \in
  \sem{\phi} \mbox{ for some } i\in [k]\}}\\[3pt]
  \bullet & \multicolumn{3}{l}{\sem{\Ki\phi}=\{\noeud\in\ltree \mid \noeuda \in
  \sem{\phi} \mbox{ for all } \noeuda \mbox{ such that
  }\noeud\reli\noeuda\}}\\[3pt]
  \bullet &  \multicolumn{3}{l}{\sem{\mu X.\phi(X)}=\biginter\{S\subseteq \ltree \mid
  \sem[{\val[S/X]}]{\phi(X)}\subseteq S\}}
\end{array}
\]
Classically, for each formula $\mu X.\phi(X)$ in $\Kmu$, the fact that
$X$ appears only positively in $\phi(X)$ ensures that
$S\mapsto\sem[{\val[S/X]}]{\phi(X)}$ is a monotone function, 
and hence that its  least fixpoint exists.  $\sem{\mu
  X.\phi(X)}$ is defined to be this fixpoint.

% In a classic way, we define dual operators: $\phi\et\psi\egdef \neg(\neg \phi
% \ou \neg \psi)$, $\Box\phi\egdef \neg\diam\neg\phi$,
% $\dKi\phi\egdef\neg\Ki\neg\phi$ and $\nu X.\phi(X)\egdef \neg \mu
% X.\neg \phi(\neg X)$.
If $\phi\in\Kmu$ is a sentence, \ie it has no
free variables,  its semantics
is independent on the valuation, that we may omit from the
semantics. For a sentence $\phi\in\Kmu$, a relation profile $\setrel$
and a tree $\ltree$, we write
$\ltree,\setrel\models\phi$ for  $\racine\in\semb{\phi}$, and we let $\lang(\phi,\setrel)\egdef \{\ltree \mid
\ltree,\setrel,\racine\models\phi\}$.  Finally, we let
  $\Lmu$ denote the sublanguage of $\Kmu$ obtained by removing the
  modalities $\Ki$, and simply write $\ltree,\epsilon\models \phi$ as
  relation profile do not play any role in the semantics of
  $\Lmu$-formulas; thus, for $\phi\in\Lmu$ we may use
  $\lang(\phi)=\{\ltree\mid \ltree,\epsilon\models \phi\}$.

% Finally, we let $\Lmu$ denote
% the sublanguage of $\Kmu$ obtained by removing from the grammar the
% rule $\phi ::= \Ki\phi$, and for a formula $\phi\in\Lmu$ we let
% $\lang(\phi)=\{\ltree\mid \ltree,\epsilon\models \phi\}$ (note that
% for this sublanguage relation profiles can be removed from the
% models).

%\sp{je me suis arrêtée là et reprends après le skype avec Catalin à 14h30.}
\subsection{Alternating-time Temporal Logic with imperfect
  information}
\label{sec-ATL}

\newcommand{\APact}{\APf_{\mbox{\scriptsize act}}}
\newcommand{\ltreea}{\ltree}

We now recall the syntax and semantics of
Alternating-time Temporal Logic with imperfect information
(\ATLi). Again, let
$\agents$ be a nonempty finite set of agents. The syntax of $\ATLi(\agents)$ is
defined by the following grammar:
\[\phi ::=  p \mid \neg \phi \mid \phi \ou \phi \mid \stratA \X
\phi \mid \stratA \phi\until\phi\]
where $p\in\AP$ and $A\subseteq\agents$. 

The semantics of \ATLi is usually defined on concurrent game
structures (see \cite{alur2002alternating}). These are transition
systems with states \labeled by valuations over some finite set of
propositions $\APf$, and  where every transition
is \labeled by a compound action $a=(a_1,\ldots,a_k)$, which is
interpreted as
Agent~$i\in\agents$ playing action $a_i$
during this transition. The imperfect information is
 usually introduced by letting each agent observe
only a subset of $\APf$, and by deciding whether agents remember 
the past during a play or not. This induces, for each agent, an equivalence relation  between finite plays. 

In order to make the comparison with epistemic \mucalc easier, we instead
define the semantics of \ATLi on  what we call tree-arenas: 
\begin{definition}
\label{def-tree-arena}
Let $\APf\subset \AP$ be a \emph{finite} set of atomic propositions, and for each $i\in\agents$, let
$\setacti$ be a nonempty finite set of actions available to Agent~$i$. Define
$\setact\egdef\bigtimes_{i\in\agents}\setacti$, and let $\APact\egdef\{p_a\mid a\in\setact\}$ where each $p_a$ is an atomic
 proposition not in $\APf$. 
An \emph{$(\APf,\setact)$-tree-arena}  is an $(\APf\union\APact)$-tree $\ltreea=(\tree,\lab)$
such that $\lab(\racine)\inter\APact=\emptyset$, and for all $\noeud\in\tree\setminus\{\racine\}$,
$\lab(\noeud)\inter\APact$ is a singleton.
\end{definition}
 For the rest of this section, we fix a finite set $\APf\subset\AP$ and a finite set of actions
 $\setacti$ for each agent $i\in\agents$. For an $(\APf,\setact)$-tree-arena
$\ltreea=(\tree,\lab)$ and a node $\noeud\in\tree$, we  write
$\lab(\noeud)=(v,a)$, where $a\in\setact$ is the unique \bmchanged{(compound)} action 
such that $p_a\in\lab(\noeud)$, and $v=\lab(\noeud)\setminus\{p_a\}$.
% \emph{valuations} over $\APf$ will be denoted as $v,v'\ldots$ and
% tuples of actions in $\setact$ by $a,a'\ldots$  Labels of nodes in
% a tree-arena are thus of the form
%  $(v,a)$, with $v\subseteq 2^{\AP}$ and $a\in\setact$.
 In addition,
 given $a=(a_1,\ldots,a_k)\in\setact$, $a^i$ will denote $a_i$.  Note
 that 
 a tree-arena $\ltreea$ can be seen as a concurrent
game structure: take a node $\noeud\in\ltreea$, and let
$(v,a)$ be its label. Node $\noeud$ can be seen as a
state of a transition system, $v$ as its label, and $a$
as the label of the only transition reaching $\noeud$. Concerning the
imperfect information, similarly to the 
previous section, we introduce agents' uncertainty by means of
binary relations $\{\reli\}_{i\in\agents}$ over
$(2^{\APf\union\APact})^*$.  
Conversely, the unfolding of every concurrent game structure with
imperfect information can be seen as a tree-arena equipped with a
relation profile. We now  adapt
 the classic semantics of ATL to our setting.

First we need a few more definitions. Fix an $(\APf,\setact)$-tree-arena $\ltreea$ and a
relation profile $\setrel$. A \emph{strategy} for Agent~$i$
is a function $\strati:\ltreea\to \setacti$, that defines the strategic
choice of Agent~$i$ in each possible situation. Because agents have
imperfect information, we classically require strategies to be
consistent with the information of the agent: if $\strati$ is a
strategy for Agent~$i$, we require that for each
$\noeud,\noeuda\in\ltreea$ such that $\noeud\reli \noeuda$,
$\strati(\noeud)=\strati(\noeuda)$ \bmchanged{(note that  strategies
  satisfying this requirement
are sometimes called \emph{uniform} strategies \cite{DBLP:journals/jancl/JamrogaA07})}.
For $A\subseteq
\agents$, we call \emph{$A$-profile} a tuple $\prof=(\strati)_{i\in A}$
where $\strati$ is a strategy for Agent~$i$, and given an $A$-profile
$\prof$ and $i\in A$, we let $\profi{i}$ denote the strategy of agent $i$
in $\prof$. The \emph{outcome} of an $A$-profile \bmchanged{$\prof$ for some $A\subseteq\agents$} is the set of
behaviours that follow the strategies in the profile, \bmchanged{defined as
follows. For a node $\noeud$ of $\ltreea$, $\out(\noeud,\prof)\subseteq
\Paths(\noeud)$ is the set of paths $\pi$ in $\ltreea$ that start in
$\noeud$ and such that for all $k\geq 0$, 
if $(v,a)$ is the label of $\pi[k+1]$, then
     $\profi{i}(\pi[k])=a^i$ for all $i\in A$}.

The semantics of an \ATLi-formula $\phi$ with atomic propositions in
$\APf$ is given with respect to an $(\APf,\setact)$-tree-arena
$\ltree=(\tree,\lab)$, a relation profile $\setrel$ and a node
$\noeud\in\ltree$:%, and it is defined inductively as follows:
\[
\begin{array}{cl}
  \bullet & \ltree,\setrel,\noeud\models p \mbox{ if }p\in v, \mbox{
    where }(v,a)=\lab(\noeud)\\[3pt]
  \bullet & \ltree,\setrel,\noeud\models \neg \phi \mbox{ if }\ltree,\setrel,\noeud\not\models\phi\\[3pt]
  \bullet & \ltree,\setrel,\noeud\models \phi\ou \psi \mbox{ if
  }\ltree,\setrel,\noeud\models\phi \mbox{ or }\ltree,\setrel,\noeud\models\psi\\[3pt]
  \bullet & \ltree,\setrel,\noeud\models\stratA \X \phi  \mbox{ if there is an
  $A$-profile $\prof$ such that:}\\[3pt]
& \hspace{1.3cm} \mbox{for all $\noeuda\in\ltree$, \bmchanged{if}
  $\noeud\reli \noeuda$ for some $i\in A$, \bmchanged{then} for all
  $\pi\in\out(\noeuda,\prof)$, }\\[3pt]
& \hspace{1.3cm}\ltree,\setrel,\pi[1]\models\phi \\[3pt]
  \bullet & \ltree,\setrel,\noeud\models\stratA \phi\until \psi  \mbox{ if there is an
  $A$-profile $\prof$ such that:}\\[3pt]
& \hspace{1.3cm} \mbox{for all $\noeuda\in\ltree$, \bmchanged{if}
  $\noeud\reli \noeuda$ for some $i\in A$, \bmchanged{then} for all  $\pi\in\out(\noeuda,\prof)$,}\\[3pt]
& \hspace{1.3cm} \mbox{there is $i\geq 0$ such that $\ltree,\setrel,\pi[i]\models \psi$, and for all $0\leq j <i$, }\ltree,\setrel,\pi[j]\models\phi
\end{array}
\]

We define the following classic \bmchanged{shorthands}:
$\true\egdef p\ou\neg p$, and $\stratA \F \phi\egdef \stratA \true
\until \phi$. Finally, for a formula $\phi\in\ATLi$,  \bmchanged{a set of
(compound) actions $\setact$} and a relation
profile $\setrel$, we let
$\lang(\phi,\setact,\setrel)\egdef\{\ltreea\mid \ltreea \mbox{ is a
}(\free(\phi),\setact)\mbox{-tree-arena s.t. } \ltreea,\setrel,\racine\models\phi\}$.

\begin{remark}
\label{rem-sem-ATL}
\bmchanged{% There are in the litterature several variants for the semantics of \ATLi, concerning the extent of the knowledge an agent
% should have  about her having a strategy to achieve some objective
% (see  for a detailed discussion
% on this matter). 
We consider here the most restrictive notion of
``having a strategy'', \ie having a strategy \emph{``de re''}
\cite{DBLP:journals/jancl/JamrogaA07}. 
% An agent has a strategy
% ``de re'' if there is a strategy that the agent  \emph{knows}
% it achieves the desired objective.
However, the
result that we prove in Section~\ref{sec-inexp} still holds
with less restrictive notions of strategies:
\emph{``de dicto''} strategies, % (the agent knows
% that she has a strategy to achieve this objective, but may not know
% what is the strategy)
 or simply uniform
strategies %(the agent may be unaware that she has a strategy to achieve the objective).
 }
\end{remark}

\subsection{Jumping tree automata}
\label{sec-jumping}

Jumping tree automata (JTA) were introduced  in
\cite{maubertFSTTCS2013,maubertphd}. 
Let $\agents$ be a finite set of agents. 
For a set $X$, $\boolp(X)$ is the set of positive boolean formulas
over $X$, \ie formulas built with elements of X as atomic
propositions and using only connectives $\ou$ and $\et$. We
also allow for formulas $\true$ and $\false$, and $\et$ has precedence
over $\ou$. Elements of $\boolp(X)$ are denoted by $\alpha,\beta\ldots$ 

\begin{definition}Let $\Dir = \{\diam,\Box\}\cup\bigunion_{i\in\agents}\{\jdiami,\jboxi\}$ be the set of
\emph{automaton directions}. A 
\emph{jumping automaton} is a tuple
$\auto=(\APf,Q,\delta,q_0,\couleur)$
where  $\APf\subset\AP$ is a finite  set of atomic propositions, $Q$ a finite set of states,
$q_0\in Q$ an initial state,  $\couleur:Q\to \setn$ a colouring function, and $\delta
: Q\times 2^{\APf} \rightarrow \boolp(\Dir\times Q)$ a transition
function.
\end{definition}

Let  $\auto$ be a JTA over $\APf$. The meaning of the \emph{jump directions} $\jdiami,\jboxi$ is given
by a relation profile $\setrel=\{\reli\}_{i\in\agents}$, where for
each $i$, $\reli\subseteq (2^\APf)^*\times(2^\APf)^*$.
The acceptance of an input tree $\ltree=(\tree,\lab)$  by  $\auto$
equipped with a relation profile $\setrel$ is 
defined on a two-player parity game between Eve (the proponent) and Adam (the opponent): let
$\ltree=(\tree,\lab)$ be an  \APtree, and let
$\auto=(\Sigma,Q,\delta,q_0,\couleur)$.  We define the game 
$\sgame{\ltree}=(V,E,\couleur',v_0)$: the set of positions
is $V=\tree\times Q \times \boolp (\Dir\times Q)$, the initial
position is $(\racine,q_0,\delta(q_0,\lab(\racine)))$, and a position
$(\noeud,q,\alpha)$ belongs to Eve if $\alpha$ is of the form $\alpha_1\vee
\alpha_2$, $[\diam,q']$ or $[\jdiami,q']$; otherwise it belongs to
Adam. %The colouring function is defined after the moves.
% When the last part of a position $(x,q,\alpha)$ is of no interest or
% is clear from the context we may write
% $(x,q,\_)$. % For example, when the game enters a node $x$ in
% state $q$, the position is always $(x,q,\delta(q,\lab(x)))$.
The  possible moves in $\sgame{\ltree}$ are the following:
\begin{alignat}{3}
& \text{$(x,q,\alpha_1 \;\op\; \alpha_2) \move (x,q,\alpha_i)$} & \text{ where 
    $\,\,\op\,\,\in\{\vee,\wedge\}$ and $i\in\{1,2\}$}\\ 
&\text{$(x,q,[\rond,q']) \move (y,q',\delta(q',\lab(y)))$} & \text{  where
$\rond \in \{\diam,\Box\}$ and $y$ is a child of $x$}\\
% & \text{$(x,q,[\stay,q']) \move (x,q',\delta(q',\lab(x)))$} &\\
% & \text{$(x,q,[\up,q']) \move (y,q',\delta(q',\lab(y)))$} & \text{  where $y$ is $x$'s parent}\\
& \text{$(x,q,[\jgeni,q']) \move (y,q',\delta(q',\lab(y)))$}  & \text{  where $\jgeni \in \{\jdiami,\jboxi\}$ and $x\reli y$}
\end{alignat}

% \begin{tabular}{ll}
% \label{eq-op} $(x,q,\alpha_1 \;\op\; \alpha_2) \move (x,q,\alpha_i)$ & where 
%     $\,\,\op\,\,\in\{\vee,\wedge\}$ and $i\in\{1,2\}$\\ 
% % %    & unless $x=\epsilon$ and $\phi_i=[\up,q']$, \\
% \label{eq-rond}  $(x,q,[\rond,q']) \move (y,q',\delta(q',\lab(y)))$ & where
%     $\rond \in \{\diam,\Box\}$ and $y$ is a child of $x$\\
% \label{eq-stay}  $(x,q,[\stay,q']) \move (x,q',\delta(q',\lab(x)))$ & \\
% \label{eq-up}  $(x,q,[\up,q']) \move (y,q',\delta(q',\lab(y)))$ & where $y$ is
%     $x$'s parent\\
% \label{eq-jgen} $(x,q,[\jgen,q']) \move (y,q',\delta(q',\lab(y)))$ & where
%     $\jgen \in \{\jdiam,\jbox\}$ and $x\rel y$
% \end{tabular}

Positions of the form $(x,q,\true)$ and  $(x,q,\false)$ are deadlocks,
winning for Eve and Adam respectively.
The colouring function $\couleur'$ of
$\sgame{\ltree}$ is inherited from the one
 of $\auto$: $\couleur'(x,q,\alpha)=\couleur(q)$.
% Most of the time the starting node $x_0$ will be the root $\epsilon$
% of the tree, and in this case we simply write $\sgame{\auto}{\ltree}{}$
% instead of $\sgame{\auto}{\ltree}{\epsilon}$. 
 A tree $\ltree$ is \emph{accepted} by
$\auto$ with relation profile $\setrel$ if Eve has a winning strategy in $\sgame{\ltree}$, and
we denote by $\lang(\auto,\setrel)$ the set of trees accepted by
$\auto$ equipped with relation profile $\setrel$. If $\auto$ is an alternating
automaton (\ie it only uses automata directions $\diam$ and $\Box$),
it needs not be equipped by a relation profile to evaluate trees, and
we write $\lang(\auto)$ for the set of trees it accepts.

\begin{remark}
In general, JTA can identify children of a given current node and send
different copies independently to each one of them. This ability is
not always needed, but quantifying (existentially or universally) over
children is sufficient.
% However, when
% studying temporal logics that cannot specify in which successor a
% property should hold, but
% can only existentially or universally quantify over successors,
% simpler models of alternating automata have been considered, which are
% sometimes called
This is the case in this work, 
reason why we have presented here a \emph{symmetric}
version of jumping tree automata, just like
 \emph{symmetric} alternating automata have sometimes been considered
(see \eg \cite{DBLP:journals/jacm/KupfermanVW00}).
\end{remark}

In the following, the size of a formula $\phi$,  written $|\phi|$, is its number of
subformulas, and the size of an automaton $\auto$, written $|\auto|$, is the size of its
transition function (\ie the sum of the sizes of formulas occuring in it).

% We classically define the
% following macros: $\true\egdef p\ou \neg p$, $\false\egdef \neg \true$, $\stratA \F \phi\egdef \stratA \true \until \phi$
% and $\stratA \G \phi\egdef \stratA \neg \F \neg \phi$.

%% file: equivalence.tex
\section{Equivalence of  jumping tree automata and epistemic \mucalc}
\label{sec-equivalence}

We show that \JTA\ and $\Kmu$ are equally expressive, as stated by the following theorem.

\begin{theorem}
\label{theo-equiv}
\hspace{1cm}
\begin{itemize}
\item For every formula $\phi\in\Kmu$, there exists a jumping automaton
$\auto_\phi$ such that for every relation profile $\setrel$,
$\lang(\phi,\setrel)=\lang(\auto_\phi,\setrel)$.
\item For every jumping automaton $\wauto$, there exists an $\Kmu$-formula $\phi_\auto$ such that for every relation profile $\setrel$, $\lang(\auto,\setrel)=\lang(\phi_\auto,\setrel)$.
\end{itemize}
Moreover, the translations are effective and linear.
\end{theorem}

The rest of this section is dedicated to the proof of Theorem~\ref{theo-equiv} and to two corollaries.

We rely on the classical equivalence between the multi-modal \mucalc, written here $\LLmu$,
and alternating tree automata, when interpreted over transition
systems: A (multi-modal, $\APf$-\labeled) \emph{transition system} is
a tuple $\sys=(\sstates,\{\sreli\}_{\bmchanged{i\in I}},\sval)$, where
$\sstates$ is a set of \emph{states}, \bmchanged{$I$ is a finite set of indices}, for each $i\in I$,
$\sreli\subseteq \sstates\times\sstates$ is a binary relation, and
$\sval:\sstates\to 2^\APf$ is a \emph{\labeling function}. 
We do not detail the semantics of the \mucalc and alternating automata
over transition systems, which is very similar to the one for trees (see \cite[Chap.\ 10]{gradel2002automata}).
%We recall the following classic result.
%(see \eg %\cite{wilke2001alternating,gradel2002automata} 
%\cite[Chap.\ 9, Chap.\ 10]{gradel2002automata} for a detailed
%exposition):
\begin{proposition}{\cite[Chap.\ 9, Chap.\ 10]{gradel2002automata}}
\label{prop-equiv-classic}
\hspace{1cm}
\begin{itemize}
\item For every formula $\phi\in\LLmu$, there exists an alternating automaton
$\auto_\phi$ that
accepts precisely the transition systems verifying $\phi$.
\item For every alternating automaton $\auto$, there exists an
  $\LLmu$-formula $\phi_\auto$ whose models are exactly the transition
  systems accepted by $\auto$.
%  $\lang(\auto)=\lang(\phi_\auto)$ over transition systems.
\end{itemize}
Moreover, the translations are effective and linear.
\end{proposition}

%\begin{remark}
Now we make observation that $\APf$-trees are connected, acyclic,
rooted transition systems with one relation. Also, an
$\APf$-tree $\ltree=(\tree,\lab)$ together with a relation profile
$\{\reli\}_{i\in\agents}$ over $(2^\APf)^*$ can be seen as a
transition system $\sys_\ltree^\setrel=(\tree,\{R\}\union
\{R_i\}_{i\in\agents},\lab)$, where $\noeud R \noeuda$ if $\noeuda$ is
a child of $\noeud$, and $\noeud R_i \noeuda$ if \spchanged{$\noeud
  \reli \noeuda$}. For a relation profile $\setrel$, we define
$\classsys{\APf} \egdef \{\sys_\ltree^\setrel\mid \ltree \mbox{ is an
  $\APf$-tree}\}$, the class of all transition systems obtained by
combining $\setrel$ with $\APf$-trees.  Now, \spchanged{two additional simple
  observations are necessary to prove Theorem~\ref{theo-equiv}:}
%\begin{enumerate}
%\item 
(1) Given a relation
profile $\setrel$, an $\Kmu$-formula on $\APf$-trees can be seen as an $\LLmu$-formula on $\classsys{\APf}$, and 
%\item 
 (2) A jumping automaton equipped with a relation profile
$\setrel$ and working on $\APf$-trees can be seen as an alternating automaton
working on $\classsys{\APf}$.
%\end{enumerate}
% \end{remark}

We now argue for Theorem~\ref{theo-equiv}: For the first point, take a
formula $\phi\in\Kmu$ and a relation profile $\setrel$. See it as an
$\LLmu$-formula over $\classsys{\APf}$. By
Proposition~\ref{prop-equiv-classic}, one can build in linear time an
alternating automaton $\auto_\phi$ that has the same language as
$\phi$ on transition systems, and therefore also when restricted to
$\classsys{\APf}$. This $\auto_\phi$, when restricted to
$\classsys{\APf}$, can be seen as a jumping automaton.  Because
$\auto_\phi$ only depends on $\phi$ and not on $\setrel$, we obtain
the desired
result. % that for every formula $\phi\in\Kmu$, there exists a jumping automaton
% $\auto_\phi$ such that for every relation profile $\setrel$,
% $\lang(\phi,\setrel)=\lang(\auto_\phi,\setrel)$.
The second point of Theorem~\ref{theo-equiv} is just dealt by rolling
back the above argumentation.

Theorem~\ref{theo-equiv} has two important corollaries. First, let
us recall some definitions and results concerning recognizable
relations and jumping  automata. \bmchanged{Let $\Sigma$ be a finite alphabet.}

\begin{definition}
A relation $\rel\;\subseteq \Sigma^*\times\Sigma^*$ is
\emph{recognizable} if there are two families of regular languages
$\langu_1,\ldots,\langu_n\subseteq \Sigma^*$ and
$\langu'_1,\ldots,\langu'_n\subseteq \Sigma^*$ such that
$\rel\;=\bigunion\limits_{i=1}^n \langu_i\times\langu'_i$.
\end{definition}

For example, epistemic relations of agents whose
memory can be represented by finite state machines are recognizable
relations (see \cite{maubertphd}). 

%Recognizable relations can be dealt with by finite-state automata:
Given a recognizable relation $\rel$, one easily shows that the
language $\{w \# w' \mid w\rel w'\}$ where $\#$ is a fresh symbol can
be accepted by a finite-state word automaton;
% given two regular languages $\langu_i$ and $\langu'_i$,
%   the binary relation $\langu_i\times\langu'_i$ can be encoded as the
%   word language $\{w \# w' \mid w \in \langu_i \text{ and } w' \in
%   \langu'_i\}$, where $\#$ is a fresh symbol; such a language is
%   accepted by a finite-state (word) automaton that checks that the first word is in
%   $\langu_i$, then reads $\#$, then checks that the second word is in
%   $\langu_i$. On this basis, a finite-state
%   automaton can be given for an arbitrary recognizable relation $\rel$; 
% Intuitively, a recognizable relation is a relation that can be
% recognized by a finite word automaton that starts by reading a first
% word, then a special symbol, then a second word, and accepts if both
% words are related.  
we let \emph{size} of $\rel$, written
$|\rel|$, is then the number of states of a minimal word automaton
that recognizes the language $\{w \# w' \mid w\rel w'\}$.

\begin{theorem}{\cite{maubertFSTTCS2013,maubertphd}}
  \label{theo-recog-jumping}
  For every jumping automaton $\auto$ equipped with a relation profile
  $\setrel$, if every relation $\reli$ in $\setrel$ is recognizable, then
  there is a two-way tree automaton $\auto_{\setrel}$ that accepts the
  same language, and such that $|\auto_{\setrel}|$ is  polynomial in
  $|\auto| + \sum\limits_{i\in\agents}|\reli|$.
\end{theorem}

Restricting attention to trees of bounded arity, we obtain
  the following two corollaries:
\begin{corollary}
  The satisfiability problem for epistemic \mucalc with recognizable
  relations is \EXPTIME-complete. % \fbox{dire que la traduction de
%     $\Kmu$ vers automates est effective et linéaire}
\end{corollary}

\begin{proof}
The upper bound follows from Theorem~\ref{theo-equiv} together with
Theorem~\ref{theo-recog-jumping} and the fact that, for trees of
bounded arity, the emptiness
problem for two-way tree automata is \EXPTIME-complete
\cite{DBLP:conf/icalp/Vardi98}. The hardness follows from
EXPTIME-hardness of the satisfiability problem for standard \mucalc.
\end{proof}

\begin{corollary}
  Epistemic \mucalc with recognizable relations \spchanged{is not more expressive than (its fragment) the \mucalc.}
%has the same   expressivity as standard \mucalc.
\end{corollary}

\begin{proof}
%   Because it contains standard \mucalc, epistemic \mucalc is clearly
%   at least as expressive as \mucalc. We now show that each epistemic
%   \mucalc formula interpreted with recognizable relations can be
%   translated into a  \mucalc formula. 
  By Propositions~\ref{prop-equiv-classic}, \bmchanged{it suffices to  show that
    for each epistemic \mucalc formula $\phi$ interpreted with
    recognizable relations, there exists an alternating tree automaton
    that accepts the models of $\phi$.} % actually we will get a
%     non-deterministic tree automaton, as a special case.
    Let $\phi\in\Kmu$,
  and let  $\setrel$ be a relation profile of recognizable relations. By
  Theorem~\ref{theo-equiv}, \spchanged{there exists} a jumping automaton
  $\auto_\phi$ such that
  $\lang(\auto_\phi,\setrel)=\lang(\phi,\setrel)$. Then, by
  Theorem~\ref{theo-recog-jumping}, there is a two-way tree automaton
  $\auto_\phi^\setrel$ such that
  $\lang(\auto_\phi,\setrel)=\lang(\auto_\phi^\setrel)$. Finally, by
  \cite{DBLP:conf/icalp/Vardi98}, there is a non-deterministic \bmchanged{(hence
  alternating)} tree
  automaton $\autob_\phi^\setrel$ such that
  $\lang(\autob_\phi^\setrel)=\lang(\auto_\phi^\setrel)$, which concludes.
% And finally,
%   because nondeterministic tree automata are particular cases of
%   alternating tree automata, by Proposition~\ref{prop-equiv-classic}
%   there is a formula $\phi'\in\LLmu$ such that
%   $\lang(\phi')=\lang(\autob_\phi^\setrel)$, \ie
%   $\lang(\phi')=\lang(\phi,\setrel)$.
\end{proof}

% The last corollary of Theorem~\ref{theo-equiv} concerns common
% knowledge. Indeed, one can express the common knowledge of a
% formula $\phi$ in epistemic
% \mucalc, with the following formula:
% \[\CK{\agents}\phi \egdef \nu X.\phi\et
% \biget\limits_{i\in\agents}\Ki X\]

% \begin{proposition}
%   \label{prop-recog-relation}
%   A relation $\rel\;\subseteq \Sigma^*\times\Sigma^*$ is recognizable
%   if, and only if, the language $\{\wa\#\wb\mid \wa\rel\wb\}$
%   is regular, where $\#\notin \Sigma$ is a fresh symbol.
% \end{proposition}

%% file: aux-inexpressive.tex
\section{Inexpressivity}
\label{sec-inexp}

\def\runaut{\rho}
\def\lra{\longrightarrow}
\def\AAA{\auto}
\def\BBB{\autob}
\newcommand{\visit}[1][\strat]{\mbox{visit}_{#1}}
\newcommand{\gameo}{\game^0}
\newcommand{\gamei}{\game^{i}}
\newcommand{\gamej}{\game^{j}}
\newcommand{\gameh}[1][h]{\game^{#1}}
\newcommand{\moveo}{\move^0}
\newcommand{\movei}{\move^i}
\newcommand{\movej}{\move^j}
\newcommand{\movek}[1][k]{\move^{#1}}

In this section we prove the non-expressibility of ATL with imperfect information within 
the epistemic \mucalc. We exhibit a formula of \ATLi and a relation
profile that has no equivalent in the epistemic \mucalc evaluated with
the same relation profile. 

Let $\APf=\{p\}$, $\agents=\{a\}$ and
$\setacti[a]=\setact=\{a_0,a_1\}$. We have $\APact=\{p_{a_0},p_{a_1}\}$. Assume that Agent $a$ is synchronous blindfold,
\ie she observes nothing but the occurence of moves. Her
indistinguishability relation 
on $(\APf,\setact)$-tree arenas is therefore $\rel\;\subseteq
(2^{\APf\union \APact})^*$,
defined by $w\rel w' \mbox{ if } |w|=|w'|$.
 Consider the formula $\stratA[a]\F p \in \ATLi(\agents)$. We prove
 that there is no formula of the epistemic \mucalc that is equivalent to $\phi$ with
 regards to the singleton relation profile $\{\rel\}$. More formally:
\begin{theorem}
\label{theo-inexpressive}
For all $\phi'\in\Kmu(\APf\union\APact,\agents)$, $\lang(\phi',\rel)\neq\lang(\stratA[a]\F p ,\setact,\rel)$.
\end{theorem}

The rest of this section is dedicated to the proof of Theorem~\ref{theo-inexpressive}.

Assume towards a contradiction that there is a formula $\phi'\in \Kmu(\APf\union\APact,\agents)$ such that
$\lang(\phi',\rel)=\lang(\stratA[a]\F p,\setact,\rel)$.
By Theorem~\ref{theo-equiv},  there is a jumping automaton
$\auto$ such that $\lang(\phi',\rel)=\lang(\auto,\rel)$. 
Let  $\auto=(\APf\union\APact,Q,\delta,q_0,\couleur)$, and let $N=|Q|+1$.%  and let $\delta$
% be the transition function of $\auto$.

We  build $2^N$ tree-arenas in which the formula $\stratA[a]\F p$
holds. 
%The difference between these trees is that 
In each of them, the objective $\F p$ is attained with a different
uniform strategy. 
We exhibit, for each tree, a winning strategy in the acceptance game of $\auto$ on that tree,
and then we employ the ``pigeon hole'' principle to show that at least two of these strategies  can be combined
into a new strategy that accepts a new tree-arena,
in which the only strategy for $a$ to ensure $\F p$  is
not uniform.

We describe the family of tree-arenas that we consider (see Figure~\ref{fig:3trees}). Concretely we
only describe finite trees, infinite trees are obtained by adding
loops on leafs and unfolding the obtained graphs.
For each $\bmchanged{i\in\{1,\ldots,2^N\}} $, the tree $\ltreea_i = (\tree_i,\lab_i)$ is such that:
\begin{enumerate}
\item The root does not verify $p$: $\lab_i(\racine)=\emptyset$
\item In $\racine$, Agent $a$ can only play $a_0$. Through this
  action she can move to $2^N+2$ different children. The first $2^N$ ones verify
  $p$, but not the  last two ones. Formally, $\tree_i \cap\setn =\{0,\ldots,2^N+1\}$. For
  readability,  we call $x_{m+1}$ the node $m$ for each $m\in \{0,\ldots,2^N+1\}$ (see Figure~\ref{fig:3trees}). For $1\leq k \leq 2^N$,
  $\lab_i(x_k) = \{p,p_{a_0}\}$, and for $k\in\{2^N+1,2^N+2\}$,
  $\lab_i(x_k)=\{p_{a_0}\}$.
\item For $1\leq k \leq 2^N+2$, node $x_k$ has exactly one child $y_k=x_k\cdot
  0$ reachable
  through $a_0$, where $p$ does not hold: for $1\leq k \leq 2^N+2$, $\lab_i(y_k)=\{p_{a_0}\}$.

\item For each $k\leq 2^N+2$, the subtree $\subt{\ltree_i}{x_k}$ is a full binary tree 
of height $N$ in which each \bmchanged{non-leaf} node $x \desc x_k$ has a left child, accessed through $a_0$,
and a right child, accessed through $a_1$. The valuations are as follows. First,
for the actions:
for $1\leq k \leq 2^N+2$ and $w\in \{0,1\}^{\leq N}$, $p_{a_c}\in\lab_i(y_k\cdot
w)$, where $c$ is the last letter of $w$. Now, for the proposition
$p$.
For each $k\in\{1,\ldots, 2^N\}$, let $w_k \in \{0,1\}^N$ be the binary
representation of $k-1$. For 
$w\in\{0,1\}^{\leq N}$, if $1\leq k\leq 2^N$, then $p\in\lab_i(y_k\cdot w)$
if and only if $w=w_k$, and if
$k\in\{2^N+1,2^N+2\}$, $p\in\lab_i(y_k\cdot w)$ if and only if
$w=w_i$.
% \item 
% \item 

% Then in the subtree $\ltree_i\restr{x_j}$, 
% the node at the end of the path labeled $w_j$ is labeled with $p_1$, and all the other nodes 
% upto level $N$ in $\ltree_i\restr{x_j}$ are labeled $\emptyset$.
% \item Both nodes $\br x_1$  and $\br x_2$ which lie at the end of the path $w_i$ which start in either $x_{2^N+1}$ or $x_{2^N+2}$
% are such that $\lab_i(x) = \{p_1\}$.
% \item All nodes $y \succ x_{2^N+1}$ or $y\succ x_{2^N+2}$, with $y\neq x_1$ and $y\neq x_2$ and $|y| \leq 2N$, 
% carry the label $\lab_i(y)=\emptyset$.
\end{enumerate}

% \begin{remark}
% Before going further, we give a couple of properties of our construction:
% \begin{enumerate}
% \item For all $i,j \leq 2^N$, $\tree_i = \tree_j$ and 
% $\ltree_i\restr{x_{2^N+1}} = \ltree_i\restr{x_i}  = \ltree_i\restr{x_{2^N+2}}$.
% \item For each $z \in \setn^*$, $|z| \leq N$, 
% there are exactly $2^N+2$ other nodes that are identically observable with $x_{2^N+1}z$ (including $x_{2^N+1}z$): these nodes are 
% $x_iz$ for all $i\leq 2^N+2$.
% \end{enumerate}
% \end{remark}

%\item $T_i(1) = a_0\times \{p_1\}, T(2) = a_0 \times \{p_2\}, T(3) = a_0 \times \emptyset$
%and $T(4) = a_1 \times \emptyset$, hence $1 \sim_a 2 \sim_a 3 \not \sim_a 4$.
%\item All successor nodes of $4$ do not bear $p_1$.
%\item For each $j \leq 2^{N+1}$, denote $w_j$ the sequence of symbols from $\{a_0,a_1\}$ which 
%contains as indices the binary representation of $j$.
%Then the subtrees at nodes $1,2,3$ are full binary infinite subtrees in which, at each node $x$,
%$a_0 = \ltree_i(x0)\restr{\Act}$ and $a_1 = \ltree_i(x1)\restr{\Act}$.
%Hence, $a$ has at exactly $2^{N+1}$ sequences of actions available 
%in all these three states. 
%\item In the subtree $\ltree_i\restr{1}$, all nodes upto level $N$ are not labeled with $p_1$,
%and all nodes at level $N+1$ are labeled with $p_1$.
%\item In both subtrees $\ltree_i\restr{2}$ and $\ltree_i\restr{3}$,
%all nodes upto level $N$ are not labeled with $p_1$.
%Moreover, on level $N+1$, the node
%which corresponds with the sequence of actions $w_i$ is labeled with $p_1$,
%and all the other nodes are not labeled with $p_1$.

Observe that for all $i,j\in\{1,\ldots,2^N\}$, $\ltreea_i$ and $\ltreea_j$ share the same
underlying tree, that we shall write $\tree$:
$\tree_i=\tree_j=\tree$. Moreover, the labellings only differ on the leafs
of $\subt{\tree}{y_{2^N+1}}$ and $\subt{\tree}{y_{2^N+2}}$. 
Remark also that, since Agent $a$ observes no atomic proposition, her uniform strategies are
simply (infinite) sequences of actions. 
Also, for each $i$ such that $1\leq i \leq 2^N$,  $\gamei$  denotes $\sgame{\ltreea_i}$, the
acceptance game of $\auto$ on $\ltreea_i$ with relation $\rel$.

\begin{figure}[!]
\begin{center}
\hspace{-40pt}$\ltree_i$

%\resizebox{\textwidth}{!}{
\resizebox{!}{.3 \textheight}{
\input{arbre1.tex}
}
\end{center}
%\bigskip
\begin{center}
\hspace{-40pt}$\ltree_j$

%\resizebox{\textwidth}{!}{
\resizebox{!}{.3 \textheight}{
\input{arbre2.tex}
}
\end{center}
%\bigskip
\begin{center}
\hspace{-40pt}$\ltree_0$

%\resizebox{\textwidth}{!}{
\resizebox{!}{.3 \textheight}{
\input{arbre3.tex}
}
\end{center}
\caption{\label{fig:3trees}
\bmchanged{The tree $\ltree_i$, the tree $\ltree_j$, and the combined
 tree $\ltree_0$}.}
\end{figure}

\begin{lemma}
\label{lem-accepti}
For all $i\in\{1,\ldots,2^N\}$, Eve has a winning strategy in $\gamei$.
\end{lemma}

\begin{proof}
Let $i\in\{1,\ldots,2^N\}$. Agent $a$ has a uniform strategy in $\gamei$ for achieving $\F p$:
it consists in playing $a_0 a_0 w_i a_0^\omega$. Therefore
$\ltreea_i,\rel,\racine\models \stratA[a]\F p$, hence  
$\ltreea_i\in\lang(\auto,\rel)$. This precisely means that Eve has a
winning strategy in $\gamei$.
\end{proof}

%which means that there exists an accepting strategy $\ltree_i$, $\runaut_i = (V_i,E_i,v_0^i)$.
\newcommand{\vkq}[1][k]{v_{#1}^q}

Let us take one winning strategy $\strat_i$ for Eve in each game
$\gamei$. 
For each $1\leq i\leq 2^N$, we define $\visit[\strat_i] : \tree \to 2^Q$, 
 which maps each node of $\tree$ to the 
set of states in which $\strat_i$ visits this node:
$\visit[\strat_i](\noeud) \egdef \{q \mid \bmchanged{\exists\pi\in\out(\strat_i), \exists i \geq 0}, \exists b\in \BBB^+(Dir\times Q)
\text{ s.t. } \bmchanged{\pi[i]= (\noeud,q,b)\}}$.
Consider, for each $1\leq i \leq 2^N$, the set    $\visit[\strat_i](y_{2^N+1})$.
Since there are at most $2^{|Q|}$ different such sets of states, and we have $2^N$ strategies  with $N = |Q|+1$,
there must exist $i \neq j$ s.t. $\visit[\strat_i](y_{2^N+1})=\visit[\strat_j](y_{2^N+1})$.
For the rest of the proof we fix such a pair $(i,j)$.
We now consider the tree-arena $\ltreea_0$ that
consists in $\ltreea_i$ where the subtree
$\subt{\ltreea_i}{y_{2^N+1}}$ is replaced with $\subt{\ltreea_j}{y_{2^N+1}}$ (see
Figure~\ref{fig:3trees}).
Let us write $\gameo$ for $\sgame{\ltreea_0}$.

% We now point out a similarity of
 % that
% allows us to significantly simplify  some notations.
Observe that the three games $\gamei$, $\gamej$ and $\gameo$ share the same set of positions:
$V^0=V^i=V^j=\tree\times Q\times \boolp(\Dir\times Q)=V$. Also, for
all $1\leq k\leq 2^N+2$, $\lab_0(y_k)=\lab_i(y_k)=\lab_j(y_k)$ ($=\{p_{a_0}\}$), that
we now write $\ell$.  Because positions of the form
$(y_k,q,\delta(q,\ell))$ play an important role in the following,
we let $\vkq\egdef (y_k,q,\delta(q,\ell))$.

%We will use the following lemma (see Appendix~\ref{app-lem-bisim} for its proof).
We first establish the following crucial lemma, which allows us to
transfer the existence of winning strategies in positions $\vkq$ from
$\gamei$ and $\gamej$ to $\gameo$ (see Appendix~\ref{app-lem-bisim} for the proof).

\newcounter{lem-bisim}
\setcounter{lem-bisim}{\value{theorem}}
\begin{lemma}
\label{lem-bisim}
\hspace{1cm}
\begin{enumerate}
\item For all $q\in Q$, for $k\neq 2^N+1$,
$\gameo,\vkq\bisim \gamei,
\vkq$, and
\item for all $q\in Q$, for $k\neq 2^N+2$,
$\gameo,\vkq\bisim \gamej,
\vkq$.
\end{enumerate}
\end{lemma}

Observe that, in $\ltreea_0$, Agent $a$ has a non-uniform strategy to
achieve $\F p$, but no uniform one. Therefore,
$\ltreea_0,\rel,\racine\not\models \stratA[a]\F p$, and thence
$\ltreea_0\notin\lang(\auto,\rel)$. \bmchanged{By definition of the acceptance
for jumping automata, Eve does not have a winning strategy in $\gameo$.}
We prove the following proposition and obtain a contradiction, which
terminates the proof of Theorem~\ref{theo-inexpressive}.

\newcounter{prop-accept0}
\setcounter{prop-accept0}{\value{theorem}}
\begin{proposition}
\label{prop-accept0}
Eve has a winning strategy in $\gameo$.
\end{proposition}
\newcommand{\Startt}{\mbox{Start}_\tree}
\newcommand{\StartG}{\mbox{Start}_{\game}}
\begin{proofsketch}
% In order to prove Proposition~\ref{prop-accept0}, we first analyze the
% games $\gameo=(V^0,E^0,C^0,v_0^0)$, $\gamei=(V^i,E^i,C^i,v_0^i)$ and
% $\gamej=(V^j,E^j,C^j,v_0^j)$: \spchanged{we establish
%   Lemma~\ref{lem-bisim}, which, together with
%   Proposition~\ref{prop-game-bisim} will provide argument for the
%   existence of a winning strategy for Eve in $\gameo$. }

We give an intuition on how a winning strategy $\strat_0$ for Eve in
  $\gameo$ can be obtained. The detail can be found in Appendix~\ref{app-strat-G0}.
  Let us define $\Startt=\{\racine,x_1,\ldots,x_{2^N+2}\}$, the two
first levels of $\tree$, and $\StartG=\{(x,q,\alpha)\in V\mid x\in
\Startt\}$. Observe that every play in $\gameo$ starts in $\StartG$,
namely, in
$v_0=(\epsilon,q_0,\delta(q_0,\lab_0(\epsilon)))$, and
  may remain in $\StartG$ for an arbitrarily long time if it keeps
  jumping without going down. Otherwise, it exits $\StartG$
  by reaching some node $y_k$, in  position $\vkq$ for some $q$. Observe also that from any position of
$\StartG$, the set of moves available in $\gameo$ and in $\gamei$ \bmchanged{(and
 in $\gamej$)} are
the same. In $\gameo$, we let Eve follow $\strat_i$ as long as
the game is in $\StartG$. If the game remains in $\StartG$ for ever,
the obtained play is an outcome of $\strati$ which is   winning for
Eve in $\gamei$. Because \bmchanged{a position has} the same colour in all
games, this play is also winning for Eve in $\gameo$. Otherwise, the
game  reaches  a position of the
  form $\vkq$. If
$k\neq 2^N+1$, 
%$\gameo,(y_k,q,\delta(q,\lab_0(y_k)))\bisim
%\gamei,(y_k,q,\delta(q,\lab_i(y_k)))$ and
because $\vkq$   has been reached  by following $\strat_i$ which is winning
in $\gamei$, $\vkq$ is a winning position for Eve in
$\gamei$. By Point~1 of Lemma~\ref{lem-bisim}, $\gameo,\vkq\bisim
\gamei,\vkq$, and
by Proposition~\ref{prop-game-bisim} \bmchanged{we obtain} that Eve also has a winning strategy  from $\vkq$ in
$\gameo$. 
If $k=2^N+1$, because
$\visit[\strat_i](y_{2^N+1})=\visit[\strat_j](y_{2^N+1})$, $\strat_j$
also visits position $\vkq[2^N+1]$, and therefore
it is a winning position for Eve in
$\gamej$. Again, by  Point 2 of
Lemma~\ref{lem-bisim}, $\gameo,\vkq\bisim\gamej,\vkq$, and by
Proposition~\ref{prop-game-bisim} Eve also has a winning strategy from
$\vkq$ in $\gameo$.
\end{proofsketch}

%% file: arbre1.tex
\begin{tikzpicture}[scale=0.8]
  \tikzstyle{every node}=[draw,circle,minimum size=.5cm]
\scriptsize

\node (root) {};
% \node[red, draw=none] (ti) [above left = 5pt of root] {$t_i$};
% \node[blue, draw=none] (tj) [above right = 5pt of root] {$t_j$};
\node (x2N) [below  = of root] {$p$};
\node[draw=none] (dots1) [left = .7cm of x2N] {$\ldots$};
\node (xk) [left = .7cm of dots1] {$p$};  
\node[draw=none] (dots2) [left = .7cm of xk] {$\ldots$};
\node (x1) [left = .7cm of dots2] {$p$};  
\node (x2N+1) [right = 2.5cm of x2N] {};  
\node (x2N+2) [right = 2cm of x2N+1] {};  
\node (y2N) [below  =  of x2N] {};
\node[draw=none] (dots3) [below = .9cm of dots1] {$\ldots$};
\node (yk) [below = of xk] {};  
\node[draw=none] (dots4) [below = .9cm of dots2] {$\ldots$};
\node (y1) [below = of x1] {};  
\node (y2N+1) [below = of x2N+1] {};  
\node (y2N+2) [below = of x2N+2] {};  

\node[draw=none] [above left = 1pt of x1] {$x_1$};
\node[draw=none] [above left = 1pt of xk] {$x_k$};
\node[draw=none] [above left = 1pt of x2N] {$x_{2^N}$};
\node[draw=none] [above right = 0pt of x2N+1.east] {$x_{2^N+1}$};
\node[draw=none] [above right = 0pt of x2N+2.east] {$x_{2^N+2}$};
\node[draw=none] [above left = 1pt of y1] {$y_1$};
\node[draw=none] [above left = 1pt of yk] {$y_k$};
\node[draw=none] [above left = 1pt of y2N] {$y_{2^N}$};
\node[draw=none] [above right = 1pt of y2N+1.east] {$y_{2^N+1}$};
\node[draw=none] [above right = 1pt of y2N+2.east] {$y_{2^N+2}$};

\path[draw, shorten >= 2pt, shorten <= 2pt, ->] (root) edge[bend right
= 10] node[draw=none,left=7pt] {$a_0$} (x1);
\path[draw, shorten >= 2pt, shorten <= 2pt, ->] (root) edge[bend right
= 10] node[draw=none,left=7pt,pos=.4] {$a_0$} (xk);
\path[draw, shorten >= 2pt, shorten <= 2pt, ->] (root) -- node[draw=none,right,pos=.3] {$a_0$}
(x2N);
\path[draw, shorten >= 2pt, shorten <= 2pt, ->] (root) edge[bend
left=10] node[draw=none,right=7pt,pos=.4] {$a_0$} (x2N+1);
\path[draw, shorten >= 2pt, shorten <= 2pt, ->] (root) edge[bend
left=10] node[draw=none,right=7pt] {$a_0$} (x2N+2);
\path[draw, shorten >= 2pt, shorten <= 2pt, ->] (x1) --
node[draw=none,right] {$a_0$}
(y1);
\path[draw, shorten >= 2pt, shorten <= 2pt, ->] (xk) --
node[draw=none,right] {$a_0$} (yk);
\path[draw, shorten >= 2pt, shorten <= 2pt, ->] (x2N) -- node[draw=none,right] {$a_0$} (y2N);
\path[draw, shorten >= 2pt, shorten <= 2pt, ->] (x2N+1) --
node[draw=none,left] {$a_0$} (y2N+1);
\path[draw, shorten >= 2pt, shorten <= 2pt, ->] (x2N+2) --
node[draw=none,left] {$a_0$} (y2N+2);

\node[draw=none] (z1) [below=2.5cm of y1] {};
\node[] (z1l) [left=.3cm of z1] {$p$};
\node[] (z1r) [right=.3cm of z1] {};

\draw[-] (y1) -- (z1l);
\draw[-] (y1) -- (z1r);

\node[draw=none] (w1) [above=.2cm of z1l] {$w_1$};

\node[] (zk) [below=2.5cm of yk] {$p$};
\node[] (zkl) [left=.3cm of zk] {};
\node[] (zkr) [right=.3cm of zk] {};

\draw[-] (yk) -- (zkl);
\draw[-] (yk) -- (zkr);

\node[draw=none] (p1) [below = .5cm of yk] {};
\node[draw=none] (p2) [below left = .3cm of p1] {};
\node[draw=none] (p3) [below right= .5cm of p2.east] {};
%\node[draw=none] (p4) [below left= .4cm of p3.west] {};

\draw  [-] plot[smooth, tension=.7] coordinates {(yk.south)
  (p1)(p2.east)(p3)(zk.north)};

\node[draw=none] (wk) [above=.2cm of zk] {$w_k$};

\node[draw=none] (z2N) [below=2.5cm of y2N] {};
\node[] (z2Nl) [left=.3cm of z2N] {};
\node[] (z2Nr) [right=.3cm of z2N] {$p$};

\draw[-] (y2N) -- (z2Nl);
\draw[-] (y2N) -- (z2Nr);

\node[draw=none] (w2N) [above=.2cm of z2Nr.north east] {$w_{2^N}$};

\node[draw=none, inner sep=0pt] (z2N+1) [below=2.5cm of y2N+1] {};
\node[red] (z2N+1a) [left=.2cm of z2N+1.east] {$p$};
\node[] (z2N+1l) [left=.3cm of z2N+1] {};
\node[] (z2N+1r) [right=.3cm of z2N+1] {};

\draw[-] (y2N+1) -- (z2N+1l);
\draw[-] (y2N+1) -- (z2N+1r);

\node[draw=none] (p5) [below = .1cm of y2N+1] {};
\node[draw=none] (p6) [below left = .15cm of p5] {};
\node[draw=none] (p7) [below right= .15cm of p6.east] {};
\node[draw=none] (p8) [below left= .15cm of p7] {};
\node[draw=none] (p9) [below right= .25cm of p8.east] {};
%\node[draw=none] (p10) [below left= .25cm of p9] {};

\draw  [-,red] plot[smooth, tension=.7] coordinates {(y2N+1.south)
  (p5)(p6.east)(p7)(p8)(p9)(z2N+1a.north)};

\node[red,draw=none] (wi) [above =.3cm of z2N+1a] {$w_i$};

\node[draw=none, inner sep=0pt] (z2N+2) [below=2.5cm of y2N+2] {};
\node[color=red] (z2N+2a) [left=.2cm of z2N+2.east] {$p$};
\node[] (z2N+2l) [left=.3cm of z2N+2] {};
\node[] (z2N+2r) [right=.3cm of z2N+2] {};

\draw[-] (y2N+2) -- (z2N+2l);
\draw[-] (y2N+2) -- (z2N+2r);

\node[draw=none] (p11) [below = .1cm of y2N+2] {};
\node[draw=none] (p12) [below left = .15cm of p11] {};
\node[draw=none] (p13) [below right= .15cm of p12.east] {};
\node[draw=none] (p14) [below left= .15cm of p13] {};
\node[draw=none] (p15) [below right= .25cm of p14.east] {};
%\node[draw=none] (p16) [below left= .25cm of p15] {};

\draw  [-,red] plot[smooth, tension=.7] coordinates {(y2N+2.south)
  (p11)(p12.east)(p13)(p14)(p15)(z2N+2a.north)};

\node[red,draw=none] (wi2) [above =.3cm of z2N+2a] {$w_i$};

% \node[draw=none] (p1) [below = .5cm of yk] {};
% \node[draw=none] (p2) [below left = .3cm of p1] {};
% \node[draw=none] (p3) [below right= .4cm of p2.east] {};
% \node[draw=none] (p4) [below left= .4cm of p3.west] {};

% \draw  [-] plot[smooth, tension=.7] coordinates {(yk.south)
%   (p1)(p2.east)(p3)(p4)(zk.north)};

% \node[draw=none] (wk) [above=.3cm of zk] {$w_k$};

\end{tikzpicture}

%% file: arbre2.tex
\begin{tikzpicture}[scale=0.8]
  \tikzstyle{every node}=[draw,circle,minimum size=.5cm]
\scriptsize

\node (root) {};
% \node[red, draw=none] (ti) [above left = 5pt of root] {$t_i$};
% \node[blue, draw=none] (tj) [above right = 5pt of root] {$t_j$};
\node (x2N) [below  = of root] {$p$};
\node[draw=none] (dots1) [left = .7cm of x2N] {$\ldots$};
\node (xk) [left = .7cm of dots1] {$p$};  
\node[draw=none] (dots2) [left = .7cm of xk] {$\ldots$};
\node (x1) [left = .7cm of dots2] {$p$};  
\node (x2N+1) [right = 2.5cm of x2N] {};  
\node (x2N+2) [right = 2cm of x2N+1] {};  
\node (y2N) [below  =  of x2N] {};
\node[draw=none] (dots3) [below = .9cm of dots1] {$\ldots$};
\node (yk) [below = of xk] {};  
\node[draw=none] (dots4) [below = .9cm of dots2] {$\ldots$};
\node (y1) [below = of x1] {};  
\node (y2N+1) [below = of x2N+1] {};  
\node (y2N+2) [below = of x2N+2] {};  

\node[draw=none] [above left = 1pt of x1] {$x_1$};
\node[draw=none] [above left = 1pt of xk] {$x_k$};
\node[draw=none] [above left = 1pt of x2N] {$x_{2^N}$};
\node[draw=none] [above right = 0pt of x2N+1.east] {$x_{2^N+1}$};
\node[draw=none] [above right = 0pt of x2N+2.east] {$x_{2^N+2}$};
\node[draw=none] [above left = 1pt of y1] {$y_1$};
\node[draw=none] [above left = 1pt of yk] {$y_k$};
\node[draw=none] [above left = 1pt of y2N] {$y_{2^N}$};
\node[draw=none] [above right = 1pt of y2N+1.east] {$y_{2^N+1}$};
\node[draw=none] [above right = 1pt of y2N+2.east] {$y_{2^N+2}$};

\path[draw, shorten >= 2pt, shorten <= 2pt, ->] (root) edge[bend right
= 10] node[draw=none,left=7pt] {$a_0$} (x1);
\path[draw, shorten >= 2pt, shorten <= 2pt, ->] (root) edge[bend right
= 10] node[draw=none,left=7pt,pos=.4] {$a_0$} (xk);
\path[draw, shorten >= 2pt, shorten <= 2pt, ->] (root) -- node[draw=none,right,pos=.3] {$a_0$}
(x2N);
\path[draw, shorten >= 2pt, shorten <= 2pt, ->] (root) edge[bend
left=10] node[draw=none,right=7pt,pos=.4] {$a_0$} (x2N+1);
\path[draw, shorten >= 2pt, shorten <= 2pt, ->] (root) edge[bend
left=10] node[draw=none,right=7pt] {$a_0$} (x2N+2);
\path[draw, shorten >= 2pt, shorten <= 2pt, ->] (x1) --
node[draw=none,right] {$a_0$}
(y1);
\path[draw, shorten >= 2pt, shorten <= 2pt, ->] (xk) --
node[draw=none,right] {$a_0$} (yk);
\path[draw, shorten >= 2pt, shorten <= 2pt, ->] (x2N) -- node[draw=none,right] {$a_0$} (y2N);
\path[draw, shorten >= 2pt, shorten <= 2pt, ->] (x2N+1) --
node[draw=none,left] {$a_0$} (y2N+1);
\path[draw, shorten >= 2pt, shorten <= 2pt, ->] (x2N+2) --
node[draw=none,left] {$a_0$} (y2N+2);

\node[draw=none] (z1) [below=2.5cm of y1] {};
\node[] (z1l) [left=.3cm of z1] {$p$};
\node[] (z1r) [right=.3cm of z1] {};

\draw[-] (y1) -- (z1l);
\draw[-] (y1) -- (z1r);

\node[draw=none] (w1) [above=.2cm of z1l] {$w_1$};

\node[] (zk) [below=2.5cm of yk] {$p$};
\node[] (zkl) [left=.3cm of zk] {};
\node[] (zkr) [right=.3cm of zk] {};

\draw[-] (yk) -- (zkl);
\draw[-] (yk) -- (zkr);

\node[draw=none] (p1) [below = .5cm of yk] {};
\node[draw=none] (p2) [below left = .3cm of p1] {};
\node[draw=none] (p3) [below right= .5cm of p2.east] {};
%\node[draw=none] (p4) [below left= .4cm of p3.west] {};

\draw  [-] plot[smooth, tension=.7] coordinates {(yk.south)
  (p1)(p2.east)(p3)(zk.north)};

\node[draw=none] (wk) [above=.2cm of zk] {$w_k$};

\node[draw=none] (z2N) [below=2.5cm of y2N] {};
\node[] (z2Nl) [left=.3cm of z2N] {};
\node[] (z2Nr) [right=.3cm of z2N] {$p$};

\draw[-] (y2N) -- (z2Nl);
\draw[-] (y2N) -- (z2Nr);

\node[draw=none] (w2N) [above=.2cm of z2Nr.north east] {$w_{2^N}$};

\node[draw=none, inner sep=0pt] (z2N+1) [below=2.5cm of y2N+1] {};
\node[blue] (z2N+1a) [right=.2cm of z2N+1.west] {$p$};
\node[] (z2N+1l) [left=.3cm of z2N+1] {};
\node[] (z2N+1r) [right=.3cm of z2N+1] {};

\draw[-] (y2N+1) -- (z2N+1l);
\draw[-] (y2N+1) -- (z2N+1r);

\node[draw=none] (p5) [below = .4cm of y2N+1] {};
\node[draw=none] (p6) [below right = .1cm of p5.south west] {};
\node[draw=none] (p7) [below left= .2cm of p6] {};
\node[draw=none] (p8) [below right= .25cm of p7] {};
%\node[draw=none] (p9) [below left= .25cm of p8.east] {};
%\node[draw=none] (p10) [below right= .15cm of p9] {};

\draw  [-,blue] plot[smooth, tension=.7] coordinates {(y2N+1.south)
  (p5)(p6)(p7)(p8)(z2N+1a.north)};

\node[blue,draw=none] (wj) [above =.3cm of z2N+1a] {$w_j$};

\node[draw=none, inner sep=0pt] (z2N+2) [below=2.5cm of y2N+2] {};
\node[blue] (z2N+2a) [right=.2cm of z2N+2.west] {$p$};
\node[] (z2N+2l) [left=.3cm of z2N+2] {};
\node[] (z2N+2r) [right=.3cm of z2N+2] {};

\draw[-] (y2N+2) -- (z2N+2l);
\draw[-] (y2N+2) -- (z2N+2r);

\node[draw=none] (p11) [below = .4cm of y2N+2] {};
\node[draw=none] (p12) [below right = .1cm of p11.south west] {};
\node[draw=none] (p13) [below left= .2cm of p12] {};
\node[draw=none] (p14) [below right= .25cm of p13] {};
% \node[draw=none] (p15) [below left= .25cm of p15.east] {};
% \node[draw=none] (p16) [below right= .15cm of p15] {};

\draw  [-,blue] plot[smooth, tension=.7] coordinates {(y2N+2.south)
  (p11)(p12)(p13)(p14)(z2N+2a.north)};

\node[blue,draw=none] (wj) [above =.3cm of z2N+2a] {$w_j$};

% \node[draw=none] (p1) [below = .5cm of yk] {};
% \node[draw=none] (p2) [below left = .3cm of p1] {};
% \node[draw=none] (p3) [below right= .4cm of p2.east] {};
% \node[draw=none] (p4) [below left= .4cm of p3.west] {};

% \draw  [-] plot[smooth, tension=.7] coordinates {(yk.south)
%   (p1)(p2.east)(p3)(p4)(zk.north)};

% \node[draw=none] (wk) [above=.3cm of zk] {$w_k$};

\end{tikzpicture}

%% file: arbre3.tex
\begin{tikzpicture}[scale=0.8]
  \tikzstyle{every node}=[draw,circle,minimum size=.5cm]
\scriptsize

\node (root) {};
% \node[red, draw=none] (ti) [above left = 5pt of root] {$t_i$};
% \node[blue, draw=none] (tj) [above right = 5pt of root] {$t_j$};
\node (x2N) [below  = of root] {$p$};
\node[draw=none] (dots1) [left = .7cm of x2N] {$\ldots$};
\node (xk) [left = .7cm of dots1] {$p$};  
\node[draw=none] (dots2) [left = .7cm of xk] {$\ldots$};
\node (x1) [left = .7cm of dots2] {$p$};  
\node (x2N+1) [right = 2.5cm of x2N] {};  
\node (x2N+2) [right = 2cm of x2N+1] {};  
\node (y2N) [below  =  of x2N] {};
\node[draw=none] (dots3) [below = .9cm of dots1] {$\ldots$};
\node (yk) [below = of xk] {};  
\node[draw=none] (dots4) [below = .9cm of dots2] {$\ldots$};
\node (y1) [below = of x1] {};  
\node (y2N+1) [below = of x2N+1] {};  
\node (y2N+2) [below = of x2N+2] {};  

\node[draw=none] [above left = 1pt of x1] {$x_1$};
\node[draw=none] [above left = 1pt of xk] {$x_k$};
\node[draw=none] [above left = 1pt of x2N] {$x_{2^N}$};
\node[draw=none] [above right = 0pt of x2N+1.east] {$x_{2^N+1}$};
\node[draw=none] [above right = 0pt of x2N+2.east] {$x_{2^N+2}$};
\node[draw=none] [above left = 1pt of y1] {$y_1$};
\node[draw=none] [above left = 1pt of yk] {$y_k$};
\node[draw=none] [above left = 1pt of y2N] {$y_{2^N}$};
\node[draw=none] [above right = 1pt of y2N+1.east] {$y_{2^N+1}$};
\node[draw=none] [above right = 1pt of y2N+2.east] {$y_{2^N+2}$};

\path[draw, shorten >= 2pt, shorten <= 2pt, ->] (root) edge[bend right
= 10] node[draw=none,left=7pt] {$a_0$} (x1);
\path[draw, shorten >= 2pt, shorten <= 2pt, ->] (root) edge[bend right
= 10] node[draw=none,left=7pt,pos=.4] {$a_0$} (xk);
\path[draw, shorten >= 2pt, shorten <= 2pt, ->] (root) -- node[draw=none,right,pos=.3] {$a_0$}
(x2N);
\path[draw, shorten >= 2pt, shorten <= 2pt, ->] (root) edge[bend
left=10] node[draw=none,right=7pt,pos=.4] {$a_0$} (x2N+1);
\path[draw, shorten >= 2pt, shorten <= 2pt, ->] (root) edge[bend
left=10] node[draw=none,right=7pt] {$a_0$} (x2N+2);
\path[draw, shorten >= 2pt, shorten <= 2pt, ->] (x1) --
node[draw=none,right] {$a_0$}
(y1);
\path[draw, shorten >= 2pt, shorten <= 2pt, ->] (xk) --
node[draw=none,right] {$a_0$} (yk);
\path[draw, shorten >= 2pt, shorten <= 2pt, ->] (x2N) -- node[draw=none,right] {$a_0$} (y2N);
\path[draw, shorten >= 2pt, shorten <= 2pt, ->] (x2N+1) --
node[draw=none,left] {$a_0$} (y2N+1);
\path[draw, shorten >= 2pt, shorten <= 2pt, ->] (x2N+2) --
node[draw=none,left] {$a_0$} (y2N+2);

\node[draw=none] (z1) [below=2.5cm of y1] {};
\node[] (z1l) [left=.3cm of z1] {$p$};
\node[] (z1r) [right=.3cm of z1] {};

\draw[-] (y1) -- (z1l);
\draw[-] (y1) -- (z1r);

\node[draw=none] (w1) [above=.2cm of z1l] {$w_1$};

\node[] (zk) [below=2.5cm of yk] {$p$};
\node[] (zkl) [left=.3cm of zk] {};
\node[] (zkr) [right=.3cm of zk] {};

\draw[-] (yk) -- (zkl);
\draw[-] (yk) -- (zkr);

\node[draw=none] (p1) [below = .5cm of yk] {};
\node[draw=none] (p2) [below left = .3cm of p1] {};
\node[draw=none] (p3) [below right= .5cm of p2.east] {};
%\node[draw=none] (p4) [below left= .4cm of p3.west] {};

\draw  [-] plot[smooth, tension=.7] coordinates {(yk.south)
  (p1)(p2.east)(p3)(zk.north)};

\node[draw=none] (wk) [above=.2cm of zk] {$w_k$};

\node[draw=none] (z2N) [below=2.5cm of y2N] {};
\node[] (z2Nl) [left=.3cm of z2N] {};
\node[] (z2Nr) [right=.3cm of z2N] {$p$};

\draw[-] (y2N) -- (z2Nl);
\draw[-] (y2N) -- (z2Nr);

\node[draw=none] (w2N) [above=.2cm of z2Nr.north east] {$w_{2^N}$};

\node[draw=none, inner sep=0pt] (z2N+1) [below=2.5cm of y2N+1] {};
\node[blue] (z2N+1a) [right=.2cm of z2N+1.west] {$p$};
\node[] (z2N+1l) [left=.3cm of z2N+1] {};
\node[] (z2N+1r) [right=.3cm of z2N+1] {};

\draw[-] (y2N+1) -- (z2N+1l);
\draw[-] (y2N+1) -- (z2N+1r);

\node[draw=none] (p5) [below = .4cm of y2N+1] {};
\node[draw=none] (p6) [below right = .1cm of p5.south west] {};
\node[draw=none] (p7) [below left= .2cm of p6] {};
\node[draw=none] (p8) [below right= .25cm of p7] {};
%\node[draw=none] (p9) [below left= .25cm of p8.east] {};
%\node[draw=none] (p10) [below right= .15cm of p9] {};

\draw  [-,blue] plot[smooth, tension=.7] coordinates {(y2N+1.south)
  (p5)(p6)(p7)(p8)(z2N+1a.north)};

\node[blue,draw=none] (wj) [above =.3cm of z2N+1a] {$w_j$};

\node[draw=none, inner sep=0pt] (z2N+2) [below=2.5cm of y2N+2] {};
\node[color=red] (z2N+2a) [left=.2cm of z2N+2.east] {$p$};
\node[] (z2N+2l) [left=.3cm of z2N+2] {};
\node[] (z2N+2r) [right=.3cm of z2N+2] {};

\draw[-] (y2N+2) -- (z2N+2l);
\draw[-] (y2N+2) -- (z2N+2r);

\node[draw=none] (p11) [below = .1cm of y2N+2] {};
\node[draw=none] (p12) [below left = .15cm of p11] {};
\node[draw=none] (p13) [below right= .15cm of p12.east] {};
\node[draw=none] (p14) [below left= .15cm of p13] {};
\node[draw=none] (p15) [below right= .25cm of p14.east] {};
%\node[draw=none] (p16) [below left= .25cm of p15] {};

\draw  [-,red] plot[smooth, tension=.7] coordinates {(y2N+2.south)
  (p11)(p12.east)(p13)(p14)(p15)(z2N+2a.north)};

\node[red,draw=none] (wi2) [above =.3cm of z2N+2a] {$w_i$};

% \node[draw=none] (p1) [below = .5cm of yk] {};
% \node[draw=none] (p2) [below left = .3cm of p1] {};
% \node[draw=none] (p3) [below right= .4cm of p2.east] {};
% \node[draw=none] (p4) [below left= .4cm of p3.west] {};

% \draw  [-] plot[smooth, tension=.7] coordinates {(yk.south)
%   (p1)(p2.east)(p3)(p4)(zk.north)};

% \node[draw=none] (wk) [above=.3cm of zk] {$w_k$};

\end{tikzpicture}

%% file: conclusion.tex
\section{Conclusions}
\label{sec-conclusion}

We have investigated in the expressive power of the epistemic \mucalc
by comparing it with jumping automata and $\ATLi$.  For the first
comparison, we have shown that, like in the classic case, $\Kmu$ is
expressively equivalent to alternating jumping tree automata.  Next,
we have shown that $\ATLi$ may express properties not expressible in
$\Kmu$, when interpreted with  synchronous  perfect-recall semantics. We have also shown that $\Kmu$ has a decidable
satisfiability problem when the semantics relies on recognizable
 relations, \ie bounded-memory semantics.

From the first two results above, one may prove that the monadic second order logic on trees,
enriched with the equal-level predicate ($\text{MSO}_{eqlevel}$) \cite{thomas-msoeqlevel}, is strictly more expressive
than $\Kmu$: on the one hand, for each jumping automaton, 
one may build an equivalent $\text{MSO}_{eqlevel}$ formula, by appropriately 
encoding Eve's winning strategies in the automaton. 
On the other hand, it is not hard to see that $\text{MSO}_{eqlevel}$ may encode any $\ATLi$ formula.  
These results strengthen the common belief that there exists no ``fixpoint'' axiomatization of $\ATLi$,
contrary to what is known for ATL with perfect information, where the
coalition  operators have 
fixpoint expansions.

We plan to further investigate the impact of these results on a theory
of jumping automata and their relation with MSO with the equal-level
predicate, or other binary predicates.  We conjecture that languages of jumping automata are not
closed under existential quantifications.  We also plan to identify a
generalization of jumping automata which would be expressively equivalent
(modulo bisimulations) to MSO with additional predicates.  On the
other hand, our non-expressiveness proof relies on the synchronous perfect recall setting, and we do not have
an easy generalization to the case of non-synchronous perfect recall
semantics, or to other types of semantics based on non-recognizable
indistinguishability relations.

%% file: appendix.tex
\section{Proof of Lemma~\ref{lem-bisim}}
\label{app-lem-bisim}
%\begin{proof}

\setcounter{theorem}{\value{lem-bisim}}
\begin{lemma}
%\label{lem-bisim}
\hspace{1cm}
\begin{enumerate}
\item For all $q\in Q$, for $k\neq 2^N+1$,
$\gameo,\vkq\bisim \gamei,
\vkq$, and
\item for all $q\in Q$, for $k\neq 2^N+2$,
$\gameo,\vkq\bisim \gamej,
\vkq$.
\end{enumerate}
\end{lemma}

\begin{proof}
For convenience, for $v,v'\in V$ and $k\in\{0,i,j\}$, we shall write
$v\movek v'$ if $(v,v')\in E^k$.  

We start with point 1.  Let us define the binary relation $Z\subseteq
V^0\times V^i$ as the smallest relation such that, for all $q\in Q$
and all $\alpha\in\boolp(\Dir\times Q)$:
\begin{itemize}
  \item $\forall k\neq 2^N+1$, $\forall x\in \subt{\tree}{y_k}$,
    $(x,q,\alpha) Z (x,q,\alpha)$,
  \item $\forall w\in\{0,1\}^*$, 
    $(y_{2^N+1}\cdot w,q,\alpha) Z (y_j\cdot w,q,\alpha)$, and
  \item $\forall w\in\{0,1\}^*$, 
    $(y_i\cdot w,q,\alpha) Z (y_{2^N+1}\cdot w,q,\alpha)$.
\end{itemize}

We prove that $Z$ is a bisimulation between $\gameo$ and
$\gamei$. Take $(v,v')\in Z$. By definition of $Z$, $v$ and $v'$ are
on the horizontal line of $y_k$ or below. Also, there are $x,x',q$
and $\alpha$ such that 
$v=(x,q,\alpha)$ and $v'=(x',q,\alpha)$.

First, for colour harmony: by definition of the colours in acceptance
games, it holds that $C(v)=C(q)=C(v')$. 

Now, for Zig, take  $u\in V$ such that $v \moveo u$. According to the
possible moves in the semantic games (see Section~\ref{sec-jumping}),
this move is of one of the three following kinds: 
\begin{enumerate}
\item it decomposes $\alpha$
without moving in the tree nor changing state, 
\item it goes down to a child of $x$ in a state $q'$, or
\item it jumps to a node $y$ such that $x\rel y$ in a state $q'$.
\end{enumerate}

{\bf Case 1}: We have $u=(x,q,\beta)$, where $\beta$ is
some subformula of $\alpha$. According to the definition of semantic
games, this move is also possible in $\gamei$: $v'\movei
u$. Therefore, we let
$u'=u$.
Because we have $(x,q,\alpha')Z(x',q,\alpha')$ for some
$\alpha'=\alpha$, by definition of $Z$, it is true for all $\alpha'$,
and in particular $(x,q,\beta) Z (x',q,\beta)$. Finally, $u Z u'$.

{\bf Case 2}: We have $\alpha=\Diamond q'$ or $\alpha=\Box q'$,
$u=(y,q',\delta(q',\lab_0(y)))$ for some child $y$ of $x$; write
$\beta\egdef \delta(q',\lab_0(y))$ and $y\egdef x\cdot c$, where $c\in\{0,1\}$.

First, observe that by definition of $Z$, $x$ and $x'$ are at the same
level ($|x|=|x'|$), and therefore if $x\cdot c$ exists in $\tree$, so does $x'\cdot
c$.  It follows, by definition of semantic games, that $v' \movei (x'\cdot c,
  q',\delta(q',\lab_i(x'\cdot c)))$ is a legal move in $\gamei$; write 
  $\beta'\egdef\delta(q',\lab_i(x'\cdot c))$ and $u'\egdef (x'\cdot c,
  q',\beta')$.

We distinguish three possibilities again, according to the definition
of $Z$ and the fact that $(x,q,\alpha) Z (x',q,\alpha)$.
\begin{itemize}
  \item  $x=x'$. We have $y=x\cdot c= x'\cdot c$. By definition of
    $Z$, we obtain that  $y\notin\subt{\tree}{y_{2^N+1}}$, so that
    $\lab_0(y)=\lab_i(y)$.
    Therefore $\beta=\beta'$, and $u=u'$, which, by definition of $Z$,
    entails that $u Z u'$.
  \item $x=y_{2^N+1}\cdot w$ for some $w$.  Because $v Z v'$,
    we have $x'=y_j\cdot w$. By observing $\ltreea_0$ and $\ltreea_i$,
    we obtain that $\lab_0(y_{2^N+1}\cdot w\cdot c)=\lab_i(y_j\cdot
    w\cdot c)$, so $\beta=\beta'$, and again, by definition of $Z$, $u
    Z u'$.
  \item  $x=y_i\cdot w$ for some $w$. Because $v Z v'$, we have  that $x'=y_{2^N+1}\cdot w$. Again, it holds that
    $\lab_0(y_i\cdot w\cdot c)=\lab_i(y_{2^N+1}\cdot w\cdot c)$,
    therefore $\beta=\beta'$, and by definition of $Z$, $u Z u'$. 
\end{itemize}

{\bf Case 3}: We have $\alpha=\jdiam q'$ or $\alpha=\jbox q'$ for some
$q'$, $u=(y,q',\beta)$ for some $x\rel y$ and
$\beta=\delta(q',\lab_0(y))$. By definition of
$Z$, $|x|=|x'|$, and because Agent $a$ is blind, the nodes
  reachable from $x$ and $x'$ through $\rel$ coincide (they  are all
  the nodes at
  the same level). We therefore hace $|x|=|x'|=|y|$. We distinguish two cases.
\begin{itemize}
\item $y\in\subt{\tree}{y_k}$ for some $k\neq 2^N+1$: since
    $|x'|= |y|$, we have that $x'\rel y$, and therefore the move $v'\movei (y,q',\delta(q',\lab_i(y)))=u'$ is
    legal in $\gamei$. Now, because $\lab_0(y)=\lab_i(y)$, $u=u'$, hence
    $u Z u'$.
\item $y\in \subt{\tree}{y_{2^N+1}}$: let $y=y_{2^N+1}\cdot w$ for
  some $w$. We have that $|y_j\cdot w|= |y_{2^N+1}\cdot w|=|y|=|x'|$,
  hence $x'\rel y_j\cdot w$, and therefore  $v \movei (y_j\cdot w, q',
  \delta(q',\lab_i(y_j\cdot w)))=u'$ is a valid move in $\gamei$. And because
  $\lab_0(y_{2^N+1}\cdot w)=\lab_i(y_j\cdot w)$,
  $\delta(q',\lab_i(y_j\cdot w))=\beta$, and therefore $u Z u'$.
\end{itemize}

For Zag, the proof is almost the same, making use of the third point in the
the definition of $Z$ instead of the second one for simulating the
moves of $\gamei$
that jump in $\subt{\tree}{y_{2^N+1}}$. So $Z$ is a bisimulation
between $\gameo$ and $\gamei$ and, clearly, %it is not hard to see that 
for all $q\in Q$, for $k\neq 2^N+1$,
$(y_k,q,\delta(q,\lab(y_k))) Z 
(y_k,q,\delta(q,\lab(y_k)))$, \ie $\vkq Z \vkq$, so that $\gameo, \vkq
\bisim \gamei, \vkq$. % (indeed, $\lab_0(y_k)=\lab_i(y_k)$).

We turn to the proof of the  second point in Lemma~\ref{lem-bisim}.

We define the following binary relation $Z'\subseteq
V^0\times V^j$, very similar to $Z$, as the smallest relation such that, for all $q\in Q$
and all $\alpha\in\boolp(\Dir\times Q)$:
\begin{itemize}
  \item $\forall k\neq 2^N+1$, $\forall x\in \subt{\tree}{y_k}$,
    $(x,q,\alpha) Z' (x,q,\alpha)$,
  \item $\forall w\in\{0,1\}^*$, 
    $(y_{2^N+2}\cdot w,q,\alpha) Z' (y_i\cdot w,q,\alpha)$, and
  \item $\forall w\in\{0,1\}^*$, 
    $(y_j\cdot w,q,\alpha) Z' (y_{2^N+2}\cdot w,q,\alpha)$.
\end{itemize}

The only difference is that now, the moves that must be avoided are
those that jump in $\subt{\tree}{y_{2^N+2}}$, which is the part that
differs between $\ltreea_0$ and $\ltreea_j$. The rest of the proof is just
the same as for the first point.
\end{proof}

\section{Proof of Proposition~\ref{prop-accept0}}
\label{app-strat-G0}

\setcounter{theorem}{\value{prop-accept0}}
\begin{proposition}
%\label{prop-accept0}
Eve has a winning strategy in $\gameo$.
\end{proposition}

\begin{proof}
We define a strategy $\strat_0$ for Eve in $\gameo$, and we prove that
it is a winning strategy. First, for each
position of the form $\vkq$, if $\vkq$ is a
winning position for Eve in $\gameo$, we pick  a
winning strategy for Eve in $(\gameo,\vkq)$ that we call
$\strat_{\vkq}$. 
Recall that $\Startt=\{\racine,x_1,\ldots,x_{2^N+2}\}$ consists in the two
first levels of $\tree$, and $\StartG=\{(x,q,\alpha)\in V\mid x\in
\Startt\}$.
Take a
  partial play $\rho$ in $\gameo$ ending in a position of  Eve.
\begin{itemize}
\item If $\rho\in \StartG^*$, $\strat_0(\rho)\egdef \strat_i(\rho)$.
\item Otherwise, there exist $\rho'$, $k$, $q$ and $\rho''$
    such that $\rho=\rho' \cdot \vkq \cdot
      \rho''$ and  $\rho'\in\StartG^*$.
  \begin{itemize}
  \item % If $k\neq 2^N+1$ (resp.\ $k=2^N+1$) and $v$ is a winning position for Eve in $\gamei$
%     (resp.\ $\gamej$), then by Lemma~\ref{lem-bisim} and
%   Proposition~\ref{prop-game-bisim}, $v$ is also a winning position
%   for Eve in $\gameo$. 
 If $\vkq$ is a winning position for Eve in $\gameo$, $\strat_{\vkq}$ is
    defined, and we let
  $\strat_0(\rho)\egdef\strat_{\vkq}(v\cdot\rho'')$.
  \item Otherwise, define $\strat_0(\rho)$ arbitrarily.
  \end{itemize}  
\end{itemize}

\begin{lemma}
\label{lem-strat0}
$\strat_0$ is winning for Eve in $\gameo$. 
\end{lemma}
Let $\pi\in\out(\gameo,\strat_0)$. If $\pi\in\StartG^\omega$,
then $\pi$ is also a play in $\gamei$ that, moreover, follows $\strat_i$, which is 
winning for Eve in $\gamei$, so $\pi$ is winning for Eve in
$\gameo$ (recall that  positions have the same colours in the
different acceptance games). 
Otherwise, there exist $\rho$, $k$, $q$ and $\pi'$
    such that $\pi=\rho \cdot \vkq \cdot
      \pi'$ and  $\rho\in\StartG^*$.
% $\rho$, $k$, $q$, and $\pi'$ such that $\pi=\rho\cdot
% (y_k,q,\delta(q,\lab(y_k))) \cdot \pi'$ and $\rho\in\StartG^*$. Write 
% $v\egdef (y_k,q,\delta(q,\lab(y_k)))$.
% Observe that $\lab_0(y_k)=\lab_i(y_k)=\lab_k(y_k)=\{p_{a_0}\}$, that we note $\ell$.
Because $\rho\cdot \vkq$ is a partial play in $\gamei$ that
  follows $\strat_i$, which is winning for Eve in $\gamei$, $\vkq$ is
  a winning position in $\gamei$.  We distinguish two cases.
\begin{itemize}
  \item $k\neq 2^N+1$: since $\vkq$ is a winning position
      for Eve in $\gamei$, by Lemma~\ref{lem-bisim} and
  Proposition~\ref{prop-game-bisim}, $\vkq$ is also a winning position
  for Eve in $\gameo$.  % Therefore, by
%     definition of $\strat_0$, $\vkq\cdot \pi'$ follows  $\strat_{\vkq}$, which
%     is winning for Eve in $\gameo$, so $\pi=\rho\cdot \vkq\cdot \pi'$ is winning for Eve.

\item $k=2^N+1$: necessarily $q\in\visit[\strat_i](y_{2^N+1})$, and
  because $\visit[\strat_i](y_{2^N+1})=\visit[\strat_j](y_{2^N+1})$,
  some outcome of $\strat_j$ in $\gamej$ visits $\vkq$, which makes $\vkq$ a
  winning position for Eve in $\gamej$. % Again, by Lemma~\ref{lem-bisim} and
%   Proposition~\ref{prop-game-bisim}, $\vkq$ is also a winning position
%   for Eve in $\gameo$, so that $\strat_{\vkq}$ is defined, and
%   $\vkq\cdot\pi'$ follows it. Like previously, $\pi'$ is winning for Eve.

\end{itemize}
In both cases, $\strat_{\vkq}$ is defined, and by
  definition of $\strat_0$, $\vkq\cdot\pi'\in\out((\gameo,\vkq),\strat_{\vkq})$.
  Because $\strat_{\vkq}$ is winning for Eve in $(\gameo,\vkq)$,
  $\vkq\cdot\pi'$ verifies the parity condition, and therefore also
  does $\pi=\rho\cdot\vkq\cdot\pi'$.
  So $\pi$ is winning for Eve, and we are done.
%This ends the proof of Lemma~\ref{lem-strat0}, and also the proof of Proposition~\ref{prop-accept0}. %, and therefore of Theorem~\ref{theo-inexpressive}.
\end{proof}

%% file: main.bbl
\begin{thebibliography}{10}

\bibitem{alur2002alternating}
R.~Alur, Th.A. Henzinger, and Orna Kupferman.
\newblock Alternating-time temporal logic.
\newblock {\em J. ACM}, 49(5):672--713, 2002.

\bibitem{DBLP:journals/jolli/BerwangerK10}
D.~Berwanger and L.~Kaiser.
\newblock Information tracking in games on graphs.
\newblock {\em Journal of Logic, Language and Information}, 19(4):395--412,
  2010.

\bibitem{bulling-atl-survey}
N.~Bulling, J.~Dix, and W.~Jamroga.
\newblock Model checking logics of strategic ability: Complexity.
\newblock In M.~Dastani, K.~V. Hindriks, and J.-J.~C. Meyer, editors, {\em
  Specification and Verification of Multi-Agent Systems}, pages 125--160.
  Springer, 2010.

\bibitem{bulling-jamroga-mu}
N.~Bulling and W.~Jamroga.
\newblock Alternating epistemic mu-calculus.
\newblock In {\em Proceedings of IJCAI'2011}, pages 109--114. IJCAI/AAAI, 2011.

\bibitem{ChatterjeeHP10strategy-logic}
Krishnendu Chatterjee, Thomas~A. Henzinger, and Nir Piterman.
\newblock Strategy logic.
\newblock {\em Inf. Comput.}, 208(6):677--693, 2010.

\bibitem{emerson1990handbook}
E.A. Emerson.
\newblock Handbook of theoretical computer science: Formal models and
  semantics, 1990.

\bibitem{gradel2002automata}
E.~Gr{\"a}del, W.~Thomas, and Th. Wilke.
\newblock {\em Automata, Logics, and Infinite Games, volume 2500 of {LNCS}}.
\newblock Springer Verlag, 2002.

\bibitem{halpern1989complexity}
J.Y. Halpern and M.Y. Vardi.
\newblock {The complexity of reasoning about knowledge and time. 1. Lower
  bounds}.
\newblock {\em Journal of Computer and System Sciences}, 38(1):195--237, 1989.

\bibitem{jamroga2006agents}
W.~Jamroga and T.~{\AA}gotnes.
\newblock What agents can achieve under incomplete information.
\newblock In {\em Proceedings of AAMAS'2006}, pages 232--234. ACM, 2006.

\bibitem{DBLP:journals/jancl/JamrogaA07}
W.~Jamroga and Th. {\AA}gotnes.
\newblock Constructive knowledge: what agents can achieve under imperfect
  information.
\newblock {\em Journal of Applied Non-Classical Logics}, 17(4):423--475, 2007.

\bibitem{janin1996expressive}
D.~Janin and I.~Walukiewicz.
\newblock On the expressive completeness of the propositional mu-calculus with
  respect to monadic second order logic.
\newblock In {\em Proceedings of CONCUR'96}, pages 263--277. Springer, 1996.

\bibitem{DBLP:journals/tcs/Kozen83}
D.~Kozen.
\newblock Results on the propositional mu-calculus.
\newblock {\em Theor. Comput. Sci.}, 27:333--354, 1983.

\bibitem{DBLP:journals/jacm/KupfermanVW00}
O.~Kupferman, M.Y. Vardi, and P.~Wolper.
\newblock An automata-theoretic approach to branching-time model checking.
\newblock {\em J. of the ACM}, 47(2):312--360, 2000.

\bibitem{maubertphd}
B.~Maubert.
\newblock {\em Logical foundations of games with imperfect information: uniform
  strategies}.
\newblock PhD thesis, Universit\'e de Rennes 1, 2014.

\bibitem{maubertFSTTCS2013}
B.~Maubert and S.~Pinchinat.
\newblock Jumping automata for uniform strategies.
\newblock In {\em {FSTTCS}'13}, pages 287--298, 2013.

\bibitem{DBLP:conf/atva/Pinchinat07}
S.~Pinchinat.
\newblock A generic constructive solution for concurrent games with expressive
  constraints on strategies.
\newblock In {\em Proceedings of ATVA'07}, pages 253--267, 2007.

\bibitem{schobbens-atl-ir2004}
P.-Y. Schobbens.
\newblock Alternating-time logic with imperfect recall.
\newblock {\em Electronic Notes in Theoretical Computer Science}, 85(2):82--93,
  2004.

\bibitem{shilov-garanina-fixpoints}
N.V. Shilov and N.O. Garanina.
\newblock Combining knowledge and fixpoints.
\newblock Technical Report Preprint n.98,
  \texttt{http://www.iis.nsk.su/files/preprints/098.pdf}, A.P. Ershov Institute
  of Informatics Systems, Novosibirsk, 2002.

\bibitem{thomas-msoeqlevel}
Wolfgang Thomas.
\newblock Infinite trees and automaton-definable relations over omega-words.
\newblock {\em Theor. Comput. Sci.}, 103(1):143--159, 1992.

\bibitem{wiebe-ATEL2003}
W.~van~der {H}oek and M.~Wooldridge.
\newblock Cooperation, knowledge, and time: Alternating-time temporal epistemic
  logic and its applications.
\newblock {\em Studia Logica}, 75(1):125--157, 2003.

\bibitem{DBLP:conf/icalp/Vardi98}
M.Y. Vardi.
\newblock Reasoning about the past with two-way automata.
\newblock In {\em Proceedings of ICALP'98}, volume 1443 of {\em Lecture Notes
  in Computer Science}, pages 628--641, 1998.

\end{thebibliography}
